\documentclass[a4paper,headings=standardclasses]{scrartcl}
\pdfoutput=1
\usepackage[utf8]{inputenc}

\usepackage{paper_preamble}

\title{An off-shell conformally invariant \\ Galilean Weyl tensor}

\author[1,a]{Quentin Vigneron}
\author[2,b]{Philip K. Schwartz}
\author[3,c]{James Read}

\affil[1]{\acronym{ENS} de Lyon, \acronym{CRAL} \acronym{UMR5574},
  Universit\'e Claude Bernard Lyon 1, \acronym{CNRS}, \par
  Lyon 69007, France}
\affil[2]{Institute of Theoretical Physics,
  Leibniz University Hannover, \par
  Appelstraße 2, 30167 Hannover, Germany}
\affil[3]{Faculty of Philosophy, University of Oxford, \par
  Oxford \acronym{OX2 6GG},  United Kingdom}
\affil[a]{\normalfont\texttt{\href{mailto:quentin.vigneron@ens-lyon.fr}
    {quentin.vigneron@ens-lyon.fr}}}
\affil[b]{\normalfont\texttt{\href{mailto:philip.schwartz@itp.uni-hannover.de}
    {philip.schwartz@itp.uni-hannover.de}}}
\affil[c]{\normalfont\texttt{\href{mailto:james.read@philosophy.ox.ac.uk}
    {james.read@philosophy.ox.ac.uk}}}

\date{}

\begin{document}

\maketitle

\begin{abstract}
We propose a manifestly conformally invariant off-shell definition of
the Weyl tensor in Galilean geometry.  We also propose definitions for
the electric and magnetic parts thereof.  For the latter to
vanish, observers with specific kinematical properties need to exist.
While this is guaranteed by the Newton--Cartan equation, it might not
be true for other Galilean invariant theories.  Therefore, we propose
that imposing the existence of observers for which the off-shell
magnetic part is zero should be a necessary condition for a Galilean
invariant theory to be called `Newtonian', in the spirit of the
`Newtonian' condition, introduced by~Trautman, in standard
Newton--Cartan gravity.

As a side result, we also show that there exists a unique Galilean
boost-invariant connection that can be built from a Galilean structure
and a choice of Coriolis field, even when the clock form is not
closed.  No extra structure is needed, contrasting with the standard
approach used to construct boost-invariant connections which
introduces a mass gauge field.
\end{abstract}

\newpage

\enlargethispage{\baselineskip}

\tableofcontents

\newpage

\section{Introduction}

In relativistic (Lorentzian) spacetime physics, a central object of
study is the Weyl tensor---i.e.,\ the trace-free part of the Riemann
tensor.  The are several (related) reasons for this: (i)~the Weyl
tensor is conformally invariant and thus carries important information
about the conformal structure of spacetime; (ii)~the Weyl tensor
encodes the gravitational degrees of freedom which are retained even
in vacuum solutions of general relativity---i.e.,\ even in the absence
of matter (when $R_{\mu\nu} = 0$); (iii)~the Weyl tensor is a crucial
object of study in the physics of gravitational waves; (iv)~the Weyl
tensor can be used to facilitate a classification of relativistic
spacetimes (the famous Petrov classification); and so forth.  All of
this is standard lore, and can be found in any good textbook treatment
of general relativity.

What has only recently begun to arouse interest, however, is how key
structures in relativistic physics (including, but not limited to, the
Weyl tensor) carry over to and behave in theories of spacetime and
gravitation other than general relativity.  For example, the question
how the notion of conformal structure carries over to the context of
Galilean gravity (i.e.,\ `non-relativistic' gravity---roughly, the $c
\rightarrow \infty$ limit of general relativity \cite{Cartan:1923_24_86,
  Friedrichs:1928, Trautman:1963, Trautman:1965,
  Dombrowski.Horneffer:1964, Kuenzle:1972, Kuenzle:1976, Ehlers:1981,
  Ehlers:2019, Malament:2012, Hartong.Obers.Oling:2023,
  Schwartz:2026}) and Carrollian gravity (i.e.,\
`ultra-relativistic' gravity---roughly, the $c \rightarrow 0$ limit of
general relativity \cite{Duval.EtAl:2014, Bergshoeff.Gomis.Longhi:2014,
  Bergshoeff.EtAl:2017, March.Read:2025}) has been picked up in
refs.~\cite{March.Read:2026, Schwartz.Read.Vigneron:2026}, and
sensible definitions of these non-Lorentzian conformal structures have
been proposed therein.

One might think that exploring the status of such structures in
spacetime geometries other than the Lorentzian setting is a mere
academic exercise, since (of course) our acutal universe is
relativistic.  However, to say this would be too fast: on the one
hand, Galilean gravity has recently found wide application in \textit{inter alia}
condensed matter phenomena (especially the fractional quantum Hall
effect, see refs.~\cite{Son:2025, Geracie.Prabhu.Roberts:2016, Wolf.Read.Teh:2022}), Ho{\v r}ava-Lifshitz gravity
\cite{Hartong.Obers:2015}, non-relativistic holography \cite{Christensen.EtAl:2014}, and non-relativistic
string theory \cite{Harmark.Hartong.Obers:2017}; on the other hand, Carrollian gravity has
recently become \emph{à la mode} in fields such as flat-space holography \cite{Bagchi.EtAl:2024}
and the physics of black holes (see, e.g., ref.~\cite{Fiorucci.EtAl:2025})---the latter
often via Carrollian string theory, on which see, e.g., ref.~\cite{Cardona.Gomis.Pons:2016}.
Given the recent wide application of these non-relativistic
geometries, there is a good case to be made that the project of
seeking to understand the structures and invariants related to them is
a worthy one.

As such, in this article we pick up on the notion of conformal
structure in the Galilean context (the Carrollian context being set
aside for another day), and in particular explore the Weyl tensor---\emph{qua}
conformal invariant---in that setting.  Although there is already
some writing on this topic---see in particular refs.~\cite{Ehlers.Buchert:2009, Dewar.Weatherall:2018,
  Dewar.Read:2020}---that work sometimes involves mathematical errors,
and in any case is not maximally general insofar as the results obtain
only on-shell---i.e., when the equations of motion are imposed.

With this in mind, our plan for the current paper is as follows.  In
\Cref{sec:Lor_Weyl}, we recall the definitions of the Weyl tensor for Lorentzian
structures, and its electric and magnetic parts.  In \Cref{sec:Gal}, we discuss
Galilean structures, their conformal transformations, and what affine
connections should be considered given such a structure.  \Cref{sec:Galilean_limit} presents
the Galilean limit of Lorentzian geometry.  Then, applying this limit
to the Lorentzian Weyl tensor, along with the material introduced in
\Cref{sec:Gal}, we provide in \Cref{sec:Gal_Weyl_conf} a conformally invariant off-shell definition
for the Galilean Weyl tensor and its electric and magnetic parts.  We
discuss the significance of these definitions in \Cref{sec:disc}, and we conclude
in \Cref{sec:concl}.

\subsection{Notations}
\label{sec:definitions}

We will employ the following notation:
\begin{itemize}
\item The coordinate components of a connection are given by $\nabla_\mu v^\nu =:
  \partial_\mu v^\nu + \Gamma^\nu_{\mu\sigma} v^\sigma$.

\item Unit timelike vector fields with respect to a Lorentzian structure
  will be denoted by uppercase letters (e.g., $U^\mu$); unit timelike
  vector fields with respect to a Galilean structure will be denoted
  by lowercase letters (e.g., $u^\mu$).  Some tensors will depend on the
  choice of timelike vector field, e.g., the Coriolis field
  $\ukappa{}_{\mu\nu}$.  These tensors should feature an upper-left index
  `$u$' indicating this dependence.  However, unless an ambiguity can
  arise, we will drop this index after the introduction of the tensor.

\item The speed of light is denoted by $c$ and is not assumed to be unity,
  in order to properly perform the Galilean limit.\footnote{We use the
    terminology `Galilean' rather than `Newtonian' for denoting the $c
    \to \infty$ limit of Lorentzian geometry to highlight the fact
    that we are in a geometric (off-shell) situation.}  We write $\cc
  := 1/c^2$.

\item In the context of the Galilean limit, we will write `\acronym{LO}' for
  `leading order'.

\item The symmetric traceless part of a two-tensor field will be denoted
  by enclosing the corresponding two indices in angular brackets.  The
  tracelessness is defined with respect to that non-degenerate metric
  which is clear from context, which in our case will be the spatial
  metric from either the Lorentzian or the Galilean structures.  For
  example, in three dimensions, denoting the metric by $h_{ij}$, we
  write $A_{\langle ij\rangle} := A_{(ij)} - \tfrac{1}{3} h^{kl} A_{kl} h_{ij}$.
\end{itemize}

\section{The Weyl tensor of Lorentzian structures}
\label{sec:Lor_Weyl}

In this section we recall, for Lorentzian structures, the definitions
of the Weyl tensor, its electric and magnetic parts, the
1+3 formalism, and conformal transformations.

\paragraph{The Weyl tensor:}

Given a Lorentzian metric $g_{\mu\nu}$, its Weyl tensor is defined as
\begin{equation}
    C^\alpha{}_{\beta\mu\nu}
        := R^\alpha{}_{\beta\mu\nu}
            - \delta^\alpha_{[\mu} R^{\vphantom{\alpha}}_{\nu]\beta}
            + g_{\beta[\mu} R_{\nu]\sigma} g^{\sigma\alpha}
            + \frac{1}{3} \delta^\alpha_{[\mu} g^{\vphantom{\alpha}}_{\nu]\beta} R_{\sigma\lambda} g^{\sigma\lambda}\,.
\end{equation}
Given a unit timelike vector field $U^\mu$, i.e., $U_\sigma U^\sigma = -1/\cc$,
and its orthogonal projector $\ubg{U}_{\mu\nu} := g_{\mu\nu} + \cc U_\mu
U_\nu$, the electric part $\uE{U\!}_{\mu\nu}$ and magnetic part
$\uH{U\!}_{\mu\nu}$ of the Weyl tensor with respect to $U^\mu$ are defined as
\begin{align}
    \label{eq:def_E}
    \uE{U\!}_{\mu\nu}
        &:= C_{\mu\alpha\nu\beta} U^\alpha U^\beta \nonumber \\
        &\:= \frac{1}{\cc} \frac{1}{2} \ubg{}^\beta{}_{\mu} \ubg{}^\alpha{}_{\nu}
            \left(R_{\alpha\beta} - \frac{1}{3} \ubg{}_{\alpha\beta} R_{\sigma\omega} g^{\sigma\omega} \right)
            -\frac{1}{2} \ubg{}_{\mu\nu} \, R_{\alpha\beta} U^\alpha U^\beta
            + R_{\mu\alpha\nu\beta} U^\alpha U^\beta \; , \\
    \label{eq:def_H}
    \uH{U\!}_{\mu\nu}
        &:= \tfrac{1}{2} \iR{U}{\epsilon}_{\mu\alpha\beta} C^{\alpha\beta}{}_{\nu\sigma} U^\sigma \nonumber \\
        &\:= \tfrac{1}{2} \iR{U}{\epsilon}_{\alpha\beta (\mu} R^{\alpha\beta}{}_{\nu)\sigma} U^\sigma \; ,
\end{align}
where $\iR{U}{\epsilon}_{\beta\mu\nu} := \epsilon_{\alpha\beta\mu\nu} U^\alpha$ and $\epsilon_{\alpha\beta\mu\nu}$ is the
spacetime volume form associated to $g_{\mu\nu}$.  The Weyl tensor can be
expressed in terms of the electric and magnetic parts as
\begin{equation}
    \label{eq:Weyl_E_H}
    C_{\mu\nu}{}^{\alpha\beta}
        = \cc 4 \left(\ubg{}^{[\alpha}{}_{[\mu} \! \uE{}^{\beta]}{}_{\nu]}
            + \cc U^{[\alpha} U_{[\mu} \! \uE{}^{\beta]}{}_{\nu]} \right)
            + \cc 2 \iR{}{\epsilon}_{\mu\nu\sigma} U^{[\alpha} \uH{}^{\beta] \sigma}
            + \cc 2 U_{[\mu} \! \uH{}_{\nu]\sigma} \iR{}{\epsilon}^{\alpha\beta\sigma} \; .
\end{equation}

\paragraph{The 1+3 formalism:}

Using the 1+3 formalism, it is possible to express the electric and
magnetic parts in terms of quantities purely orthogonal\footnote{Since
  we may freely identify the tangent and cotangent spaces via the
  non-degenerate Lorentzian metric $g_{\mu\nu}$, the notion of being
  orthogonal to $U^\mu$ is well-defined for both vectors and
  covectors, and hence for (indices of) both contra- and covariant
  tensors.} to $U^\mu$, called \emph{kinematical quantities}.  We describe this approach
in the following.  For a complete introduction on the 1+3 formalism,
we direct the reader to ref.~\cite{Roy:2014}.

The kinematical (1+3) quantities of $U^\mu$ are defined as follows:
\begin{align}
    \text{the acceleration:}  &&\ua{U}^\mu
    &:= U^\alpha \nabla_\alpha U^\mu \; , \\
    \text{the expansion tensor:}  &&\uTheta{U\!}_{\mu\nu}
    &:= \ubg{}^\alpha{}_{(\mu} \ubg{}^\beta{}_{\nu)} \nabla_\alpha U_\beta
    = \frac{1}{2} \Lie{U} \ubg{}_{\mu\nu} \; , \\
    \text{the vorticity tensor:}  &&\uOmega{U\!}_{\mu\nu}
    &:= \ubg{}^\alpha{}_{[\mu} \ubg{}^\beta{}_{\nu]} \nabla_\alpha U_\beta \; , \\
    \text{the expansion rate:}  &&\utheta{U\!}
    &:= \uTheta{}^\alpha{}_\alpha \; , \\
    \text{the shear tensor:}  &&\usigma{U\!}_{\mu\nu}
    &:= \uTheta{}_{\langle\mu\nu\rangle} \; ,
\end{align}
where ${\langle\mu\nu\rangle}$ denotes the traceless part with respect to the spatial
projector $\ubg{}_{\mu\nu}$, i.e., $\uTheta{}_{\langle\mu\nu\rangle} := \uTheta{}_{\mu\nu} -
\tfrac{1}{3} \utheta{} \ubg{}_{\mu\nu}$.

We define the `spatial connection' $\iR{U}{D}_\mu$ acting on tensor fields
$A^{\mu \dots}{}_{\nu \dots}$ orthogonal to $U^\mu$ as the operator
\begin{align}
    \label{eq:def_D_spatial}
    \iR{U}{D}_\alpha A^{\mu \dots}{}_{\nu \dots}
        &:= \ubg{}^\omega{}_\alpha\, \ubg{}^\mu{}_\sigma \, \dots \, \ubg{}_\nu{}^\lambda \, \dots \,
    \nabla_\omega A^{\sigma \dots}{}_{\lambda \dots} \; .
\end{align}
Following Section~5.1.1 in ref.~\cite{Roy:2014}, we define the `spatial Riemann
tensor' $\iR{U\,}{\CR}^\alpha{}_{\beta\mu\nu}$ as
\begin{align}
  \omega_\alpha \, w^\beta \, v_1 ^\mu \, v_2^\nu \, \iR{U\,}{\CR}^\alpha{}_{\beta\mu\nu}
    := \omega_\alpha \left[\left(D_{v_1} D_{v_2} w - D_{v_2} D_{v_3} w \right)^\alpha
    - \ubg{}^\alpha{}_\sigma (\nabla_{[v_1,v_2]} w)^\sigma \right]
     ,
\end{align}
where $\omega^\mu$, $w^\mu$, $v_1^\mu$, and $v_2^\mu$ are all orthogonal to $u^\mu$.

\begin{remark}
  \label{rem:D_spatial}
  If the vorticity vanishes, i.e., $U_\mu$ is foliation forming, then
  $D_\mu$ is the Levi-Civita connection of the spatial Riemannian metric
  induced by $g_{\mu\nu}$ on each of the spatial leaves integrating the
  distribution $\ker U_\mu \;$.

  If the vorticity does \emph{not} vanish, $\Omega_{\mu\nu} \ne 0$, the operator $D_\mu$
  may \emph{not} be understood as a vector bundle connection (covariant
  derivative operator): it does not satisfy the Leibniz rule for
  differentiation in direction of $U^\mu$ (since this yields $0$ by
  definition), so it cannot be a connection on the vector bundle $\ker
  U_\mu$ over the spacetime manifold; further, since no spatial leaves
  exist, it cannot be understood as an affine connection on these
  (i.e., as a covariant derivative on the leaves' tangent bundles).
  In particular, it cannot be understood as the Levi-Civita connection
  of a `Riemannian spatial metric'.  In particular, its `torsion',
  defined by
  \begin{equation}
    T[D]^\alpha{}_{\mu\nu} \, \omega_\alpha \, v_1 ^\mu \, v_2^\nu
    := \omega_\alpha \left[\left(D_{v_1} w - D_{v_2} w\right)^\alpha - [v,w]^\alpha
    \right]
  \end{equation}
  for $\omega^\mu$, $v_1^\mu$, $v_2^\mu$ orthogonal to $u^\mu$ (see Section~4.2.3
  in ref.~\cite{Roy:2014}), is not purely spatial and given by
  \begin{equation}
    T[D]^\alpha{}_{\mu\nu}
    = -\cc\, 2 \, \Omega_{\mu\nu} U^\alpha\;.
  \end{equation}
  In particular, the presence of this torsion leads to an
  antisymmetric part of the `spatial Ricci tensor' (defined by
  contraction of $\CR^\alpha{}_{\beta\mu\nu}$):
  \begin{equation}
    \CR_{[\mu\nu]} = \cc \left[ \theta \, \Omega_{\mu\nu} + 2 \Theta^\alpha{}_{[\mu} \Omega_{\nu]\alpha} \right] .
  \end{equation}
  Therefore, this tensor does not enjoy the properties of the Ricci
  tensor defined from a Riemannian metric.

  However, as will be shown in \Cref{sec:limit_kin_vars}, at leading order in the Galilean
  limit, $D_\mu$ corresponds to the Levi-Civita connection of the
  Galilean spatial (Riemannian) metric, and the Ricci tensor is
  symmetric.
\end{remark}

With respect to the kinematical variables, the electric part of the
Weyl tensor takes the form (see Section~4.7.2 in ref.~\cite{Ellis.Maartens.MacCallum:2012})
\begin{equation}
    \label{eq:E_kin_vars_old}
    E_{\mu\nu}
        = \frac{1}{\lambda} \frac{1}{2} R_{\alpha\beta}  \ubg{}^{\alpha}{}_{\langle\mu} \ubg{}^\beta{}_{\nu\rangle}
        - \Lie{U} \sigma_{\langle \mu\nu \rangle}
        + \sigma_{\alpha\langle\mu} \sigma_{\nu\rangle}{}^\alpha
        + \Omega_{\alpha\langle\mu} \Omega_{\nu\rangle}{}^\alpha
        + 2 \sigma_{\alpha\langle\mu} \Omega_{\nu\rangle}{}^\alpha
        + D_{\langle\mu} a_{\nu\rangle}
        + \cc a_{\langle \mu} a_{\nu \rangle}\; .
\end{equation}
We want to express the first term in that formula in terms of the
spatial curvature.  By Corollary S3 in ref.~\cite{Roy:2014}, we have
\begin{equation}
    \frac{1}{\lambda} R_{\alpha\beta} \ubg{}^{\alpha}{}_{\langle\mu} \ubg{}^\beta{}_{\nu\rangle}
    = \CR_{\langle\mu\nu\rangle}
        + \Lie{U} \sigma_{\langle\mu\nu\rangle}
        + \frac{1}{3} \theta \sigma_{\mu\nu}
        - 2 \sigma^\alpha{}_{\langle\mu} \sigma_{\nu\rangle\alpha}
        - 2  \sigma_{\alpha\langle\mu}\Omega_{\nu\rangle}{}^\alpha
        - D_{\langle\mu} a_{\nu\rangle}
        - \cc a_{\langle\mu} a_{\nu\rangle} \; .
\end{equation}
Then, the electric part~\eqref{eq:E_kin_vars_old} can be expressed in terms of the spatial
curvature as
\begin{equation}
  \label{eq:E_kin_vars}
  E_{\mu\nu}
    = \frac{1}{2}\left[\frac{1}{\cc} \CR_{\langle\mu\nu\rangle}
      - \Lie{U} \sigma_{\langle \mu\nu \rangle}
      + \frac{1}{3} \theta \sigma_{\mu\nu}
      + 2 \left(\Theta_{\alpha\langle\mu} + \Omega_{\alpha\langle\mu}\right) \Omega_{\nu\rangle}{}^\alpha
      + D_{\langle\mu} a_{\nu\rangle}
      + \cc \, a_{\langle \mu} a_{\nu \rangle}
    \right] .
\end{equation}
The magnetic part is (see Section~4.7.2 in ref.~\cite{Ellis.Maartens.MacCallum:2012})
\begin{equation}
  \label{eq:H_kin_vars}
  H_{\mu\nu}
    = \epsilon_{\alpha\beta\langle\mu}
    \left[
      D^\alpha \left(\sigma^\beta{}_{\nu\rangle}
        + \Omega^\beta{}_{\nu\rangle}\right)
      - \cc \, a_{\mu\rangle} \Omega^{\alpha \beta}
    \right] .
\end{equation}

\paragraph{Conformal transformations:}

A conformal transformation of a Lorentzian metric $g_{\mu\nu}$ is defined
as
\begin{equation}
  g_{\mu\nu} \overset{\varphi}{\rightarrow}  \e^{2\varphi} g_{\mu\nu} \; ,
\end{equation}
where $\varphi$ is a scalar field.  Under this transformation, the 1+3
quantities transform as
\begin{subequations}
\label{eq:conf_Lor_kin_vars}
\begin{align}
  &\ U^\mu
    \overset{\varphi}{\rightarrow} \e^{-\varphi} \, U^\mu \; ,
  &&\: \: \theta
    \overset{\varphi}{\rightarrow} \e^{-\varphi} \left(\theta  + 3 U^\alpha \partial_\alpha \varphi \right) ,
  &&\sigma_{\mu\nu}
    \overset{\varphi}{\rightarrow} \e^{\varphi} \, \sigma_{\mu\nu} \; , \\
  &\Omega_{\mu\nu}
    \overset{\varphi}{\rightarrow} \e^{\varphi} \, \Omega_{\mu\nu} \; ,
  &&a_{\mu}
    \overset{\varphi}{\rightarrow} a_{\mu} + \frac{1}{\cc} \ubg{}^\alpha{}_\mu \partial_\alpha \varphi \; .
  &&
\end{align}
\end{subequations}
The Weyl tensor, its electric part, and its magnetic part are
conformally invariant.

\section{Galilean structures, connections, and conformal transformations}
\label{sec:Gal}

In this section we recall the definitions of Galilean structures,
their conformal transformations, and the decomposition theorem of any
affine connection with respect to these structures.  Then, we derive
the connections which can be constructed from a Galilean structure
alone.  For the basics of Galilean geometry, see, e.g., refs.~\cite{Malament:2012, Hartong.Obers.Oling:2023, Schwartz:2026}.

\subsection{Galilean structures}
\label{sec:Galilean_structures}

A \emph{Galilean structure} is a pair $(\tau_\mu, h^{\mu\nu})$ where $\tau_\mu$ is a
non-vanishing 1-form called the \emph{clock form}, and $h^{\mu\nu}$ is a
positive-semidefinite symmetric (2,0)-tensor field of rank 3,
satisfying $\tau_\mu h^{\mu\nu} = 0$, called the \emph{spatial metric}.

A vector (field) $v^\mu$ is said to be \emph{spacelike} if $\tau_\alpha v^\alpha = 0$, and
similarly for higher-degree contravariant tensors.  A vector (field)
$u^\mu$ is said to be \emph{unit timelike} with respect to the Galilean
structure if $\tau_\alpha u^\alpha = 1$, and (future-directed) \emph{timelike} if $\tau_\alpha u^\alpha
> 0$.  For the rest of the paper we will always consider unit timelike
vector fields, therefore we will drop `unit' for simplicity.

The field $h^{\mu\nu}$ induces a positive-definite bundle metric on the
distribution $\ker\tau_\mu$ of spacelike vectors, i.e., a proper `spatial
metric'; see, e.g., Construction~2.5 in ref.~\cite{Schwartz:2026}.

The \emph{spatial projector} along a timelike vector field $u^\mu$ is
\begin{equation}
  \ub{u}^\mu{}_\nu := \delta^\mu_\nu - u^\mu \tau_\nu \; .
\end{equation}
The \emph{covariant spatial metric} $\ub{u}_{\mu\nu} = \ub{u}_{\nu\mu}$ with respect
to $u^\mu$ is (uniquely) defined by
\begin{align}
  \ub{u}_{\nu\alpha} h^{\alpha\mu} = \ub{u}^\mu{}_\nu \; , \quad
  \ub{u}_{\mu\alpha} u^\alpha = 0 \; .
\end{align}
We use the letter $\ub{}$ rather than $\ubg{}$ to distinguish these
objects from the orthogonal projector defined for Lorentzian
structures.

While a Galilean structure alone does not define a dual relation
between forms and vectors, a standard convention is to define a
(non-invertible) way of raising and lowering indices via contraction
with the spatial metric $h^{\mu\nu}$ and the covariant spatial metric
$\ub{}_{\mu\nu}$, given a fixed choice of timelike vector field $u^\mu$.  In
this sense raising indices only depends on the Galilean structure,
while lowering indices depends on a choice of timelike vector field.
Since for most of this paper that vector field will be fixed, we can
raise and lower indices without ambiguity.

We introduce
\begin{equation}
    \omega_{\mu\nu} := \tfrac{1}{2} (\dd\tau)_{\mu\nu} = \partial_{[\mu} \tau_{\nu]} \; .
\end{equation}
For a general Galilean structure, there is no constraint on $\omega_{\mu\nu}$.
In particular, in general we have $\tau\wedge\dd\tau \ne 0$, implying that the
distribution $\ker\tau_\mu$ of spacelike vectors is not integrable.  For
the present paper we will however only consider clock forms satisfying\footnote{In the literature, this condition is sometimes called the condition of \emph{absolute simultaneity} or the \emph{twistless-torsional Newton--Cartan} condition.}
\begin{subequations}
\begin{equation}
  \tau \wedge \dd\tau = 0 \, ,
\end{equation}
such that the spacelike distribution is integrable.  We call the
foliation of the spacetime manifold integrating the spacelike
distribution $\ker\tau_\mu$ the \emph{spatial foliation}, or the foliation by \emph{spatial
  leaves}.  The bundle metric induced by $h^{\mu\nu}$ on $\ker\tau_\mu$
is then a (positive-definite) Riemannian metric on each spatial leaf
$\Sigma$, which we denote by $h_{ij}$ and call the \emph{spatial metric} as well
(in the following, no ambiguity will arise regarding whether we speak
of $h^{\mu\nu}$ or $h_{ij}$).  On top of complying with most of the
literature on Galilean structures, the hypothesis of $\tau_\mu$ being
foliation forming is also imposed by global hyperbolicity if we
require the Galilean structure to arise from a Galilean limit of a
Lorentzian metric (see \Cref{sec:ansatz_limit} and ref.~\cite{Vigneron.Barzegar.Read:2025}).  The condition $\tau \wedge \dd\tau =
0$ is equivalent to
\begin{equation}
  h^{\mu\alpha} h^{\nu\beta} \omega_{\alpha\beta} = 0 \, .
\end{equation}
Then, choosing any timelike vector field $u^\mu$, $\omega_{\alpha\beta}$ is fully
determined by the vector field
\begin{align}
  \uator{}^\mu
  &:= 2\,u^\alpha\omega_{\alpha\beta} h^{\beta\mu}
    = h^{\mu\alpha} \Lie{u} \tau_\alpha \; , \\
  \shortintertext{with}
  \omega_{\mu\nu}
  &= \tau_{[\mu} \ator_{\nu]}
    := \tau_{[\mu} \ub{}_{\nu]\alpha} \ator^{\alpha} \, .
\end{align}
\end{subequations}
Both $\ator^\mu$ and $\tau_{[\mu} \ator_{\nu]}$ do not depend on the choice of
timelike vector field.

The vector field $\ator^\mu$ is sometimes called the `torsion vector'
(see, e.g., ref.~\cite{Hansen.Hartong.Obers:2020}).\footnote{In \cite{Hansen.Hartong.Obers:2020}, the convention for unit timelike vector fields is sign-opposite to ours, i.e., $\tau_\alpha u^\alpha = -1$.  Therefore, their definition of the torsion form is opposite to ours.} It is indeed directly related to the
torsion of a connection compatible with the Galilean structure
(see~\eqref{eq:Gal_Q_T}).  However, one could consider a torsion-free non-metric
connection even in the case $\ator^\mu \ne 0$.  In this sense, it could
be also be called the `non-metricity vector'.  But because, regardless
of the choice of connection, $\ator^\mu$ encodes time dilation, we think
that a more appropriate name is the \emph{time dilation vector field}.

Finally, assuming the spacetime manifold is orientable and fixing an
orientation, we can define the \emph{spacetime volume form} of a Galilean
structure.  It may be defined geometrically as the unique 4-form $\eta$
that takes the value $1$ when evaluated on any positively oriented
\emph{Galilean basis} $(\e_t, \e_1, \e_2, \e_3)$, i.e., a positively oriented
basis of the tangent space at $p$ that is adapted to the Galilean
structure in the sense of satisfying $\tau(\e_t) = 1$ and $h|_p =
\sum_{a=1}^3 \e_a \otimes \e_a$.  In coordinate component notation, in
positively oriented coordinates adapted to the spatial foliation, the
spacetime volume form is given by
\begin{align}
  \eta_{\alpha\beta\mu\nu}
  := |\tau_0| \left(\det\left(h^{ij}\right)\right)^{-1/2} \varepsilon_{\alpha\beta\mu\nu} \; ,
\end{align}
where $\varepsilon_{\alpha\beta\mu\nu}$ is the antisymmetric Levi-Civita symbol, and
$\det\left(h^{ij}\right)$ is the determinant of the matrix of spatial
components of $h^{\mu\nu}$.  Given a choice of timelike vector field
$u^\mu$, the \emph{spatial volume form} $\ueta{u}_{\alpha\mu\nu}$, which is annihilated
by that vector field, is defined as
\begin{align}
  \ueta{u}_{\alpha\mu\nu} := u^\beta \, \eta_{\beta\alpha\mu\nu} \; .
\end{align}

\begin{remark}\label{rem:proj}
  Any spacelike tensor can be pulled back uniquely to the spatial
  leaves.  We call these leaves $\Sigma$, with the spacetime manifold
  taking the product form $\mR\times\Sigma$.  Equivalently, for any choice of
  timelike vector field $u^\mu$, any tensor annihilated by it can be
  pulled back uniquely on $\Sigma$ (see ref.~\cite{Vigneron:2021}).  In other words,
  spacelike tensors or tensors annihilated by $u^\mu$ are fully
  determined by their spatial components in a basis adapted to the
  spatial foliation.  These spatial components define a unique tensor
  on each $\Sigma$.  We will use this property to express the electric and
  magnetic Galilean Weyl tensors on $\Sigma$ in \Cref{sec:Weyl_spatial_approach}.
\end{remark}

\subsection{Galilean conformal structures and connections}
\label{sec:Gal_con_conf}

In this section, we define Galilean conformal transformations,
conformal Galilean structures, and their compatible connections (see
ref.~\cite{Schwartz.Read.Vigneron:2026} and references therein for a detailed discussion of these
definitions).  These definitions will be important when assessing the
conformal invariance of the Galilean Weyl tensor in \Cref{sec:Gal_Weyl_conf}.

A \emph{Galilean conformal transformation} of Galilean structures is defined as
\begin{equation}
  \label{eq:conf_trans_Gal}
  \tau_\mu \overset{\varphi}{\rightarrow} \e^{\conf} \tau_\mu \; , \quad
  h^{\mu\nu} \overset{\varphi}{\rightarrow} \e^{-2\conf} h^{\mu\nu} \; ,
\end{equation}
where $\conf$ is a scalar field.  We can extend this definition to any
timelike vector field $u^\mu$, its covariant spatial metric $\ub{}_{\mu\nu}$,
and its Coriolis field $\ukappa{}_{\mu\nu}$ (see \Cref{sec:Gal_con} for its definition) as
\begin{equation}
  u^\mu \overset{\varphi}{\rightarrow} \e^{-\conf} u^\mu \; , \quad
  \ub{}_{\mu\nu} \overset{\varphi}{\rightarrow} \e^{2\conf} \ub{}_{\mu\nu} \; , \quad
  \ukappa{}_{\mu\nu} \overset{\varphi}{\rightarrow} \e^{\conf} \ukappa{}_{\mu\nu} \; .
\end{equation}
Note that, in general, Galilean conformal transformations~\eqref{eq:conf_trans_Gal} do not
preserve closedness of the clock form, but preserve the condition
$\tau \wedge \dd\tau = 0$.  Since in the present paper we want to define a
conformally invariant Galilean Weyl tensor, it is therefore essential
to allow for non-closed, foliation-forming clock forms.

A \emph{conformal Galilean structure}, denoted $[\tau, h]$, is an equivalence
class of Galilean structures under Galilean conformal transformations.
A connection $\nabla$ is said to be \emph{compatible} with a conformal Galilean
structure $[\tau, h]$, or \emph{conformally compatible}, if for some
representative $(\tau_\mu, h^{\mu\nu})$ of the conformal structure there exists
a 1-form $n_\mu$ such that
\begin{align}
  \label{eq:conf_connection}
  \nabla_\mu \tau_\nu = n_\mu \tau_\nu\;, \quad
  \nabla_\alpha h^{\mu\nu} = - 2 n_\alpha h^{\mu\nu}\;.
\end{align}
Under a Galilean conformal transformation of the representative, the
1-form $n_\mu$ transforms as $n_\mu \rightarrow n_\mu + \partial_\mu \conf$.

\subsection{Galilean structures and general affine connections}
\label{sec:Gal_con}

Given a Lorentzian metric $g_{\mu\nu}$, a general affine connection $\nabla$ is
uniquely determined by its non-metricity $\nabla_\alpha g_{\mu\nu}$ and its torsion
$T^\alpha{}_{\mu\nu} = 2\Gamma^\alpha_{[\mu\nu]}$, which may be freely specified.  The unique
connection with vanishing torsion and non-metricity is the Levi-Civita
connection of the metric.  This gives rise, in particular, to the
notions of \emph{distortion}, \emph{contortion} and \emph{disformation}: the distortion is
the difference of a general connection to the Levi-Civita one, the
contortion is its torsion-dependent part, and the disformation is the
non-metricity-dependent part.

A similar situation appears for Galilean structures.  As shown in
ref.~\cite{Schwartz:2025}, given a Galilean structure $(\tau_\mu, h^{\mu\nu})$ and a choice
of timelike vector field $u^\mu$, a general affine connection is
uniquely determined by the non-metricities $Q_{\mu\nu} := \nabla_\mu \tau_\nu$ and
$Q_\alpha{}^{\mu\nu} := \nabla_\alpha h^{\mu\nu}$, the torsion $T^\alpha{}_{\mu\nu} := 2\Gamma^\alpha_{[\mu\nu]}\;$,
and the Coriolis (2-form) field with respect to $u^\mu$ defined as
$\ukappa{u}_{\mu\nu} := (\nabla_{[\mu} u^\alpha) \ub{u}_{\nu]\alpha}\;$: it takes the form
\begin{subequations} \label{eq:con_general}
\begin{align}
  \Gamma^\alpha_{\mu\nu}
  &= \ucheckGamma{u}^\alpha_{\mu\nu}
    + 2\tau_{(\mu} \ukappa{}_{\nu)}{}^\alpha
    + \ub{}^\alpha{}_\beta Q_{(\mu\nu)}{}^\beta
    - \tfrac{1}{2} Q^\alpha{}_{\mu\nu}
    - u^\alpha Q_{(\mu\nu)}
    - T_{(\mu\nu)}{}^\alpha
    + \tfrac{1}{2} T^\alpha{}_{\mu\nu} \; , \\
  \intertext{where indices have been lowered with $\ub{u}_{\mu\nu}$, we set}
  \label{eq:con_u}
  \ucheckGamma{u}^\alpha_{\mu\nu}
  &:= u^\alpha \partial_{(\mu} \tau_{\nu)}
    + h^{\alpha\sigma} \bigl( \partial_{(\mu} \ub{}_{\nu)\sigma}
      - \tfrac{1}{2} \partial_\sigma \ub{}_{\mu\nu} \bigr) \, ,
\end{align}
\end{subequations}
and where the non-metricities are constrained by
\begin{equation}
  \label{eq:Gal_Q_T}
  \tau_\sigma Q_\mu{}^{\nu\sigma} = -Q_{\mu}{}^\nu \; , \quad
  \tau_\sigma T^\sigma{}_{\mu\nu} = -2Q_{[\mu\nu]} + 2\omega_{\mu\nu} \; .
\end{equation}
Conversely, the fields $Q_{\mu\nu}$, $Q_\alpha{}^{\mu\nu}$, $T^\alpha{}_{\mu\nu}$, $\kappa_{\mu\nu}$
may be freely specified, up to the constraints \eqref{eq:Gal_Q_T}, and then the
connection defined by \eqref{eq:con_general} has these fields as its non-metricities,
torsion, and Coriolis field with respect to $u^\mu$.

\enlargethispage{.4\baselineskip}

In other words, $Q_\alpha{}^{\mu\nu}$, $Q_{\mu\nu}$, $T^\alpha{}_{\mu\nu}$ and
$\ukappa{}_{\mu\nu}$ are the `building blocks' of any affine connection
with respect to a Galilean structure, similarly to $\nabla_\alpha g_{\mu\nu}$ and
$T^\alpha{}_{\mu\nu}$ being the building blocks with respect to a Lorentzian
metric.

Fixing the connection, one can compute its Coriolis field with respect
to a different timelike vector field $u^\mu - v^\mu$ (with $v^\mu$
spacelike): it is given by
\begin{align}
  \label{eq:fix_con_Coriolis_trafo}
  \ukappa{u - v}_{\mu\nu}
  &= \ukappa{u}_{\mu\nu}
    - \partial_{[\mu} v_{\nu]}
    + \tau_{[\mu} \partial_{\nu]} \left(\tfrac{v^2}{2}\right)
    + \tfrac{1}{2} T^\rho{}_{\mu\nu} v_\rho \nonumber \\
  &\quad- Q_{[\mu\nu]\rho} v^\rho
    + \tfrac{1}{2} \tau_{[\mu} Q_{\nu]\rho\sigma} v^\rho v^\sigma
    + (u^\rho - v^\rho) \bigl(v_{[\mu} + v^2 \tau_{[\mu}\bigr) Q_{\nu]\rho} \; ,
\end{align}
where indices have been lowered with $\ub{u}_{\mu\nu}$, and $v^2 = v^\alpha
v_\alpha$.
\pagebreak

We see from~\eqref{eq:Gal_Q_T} that, if $\dd \tau \ne 0$, any metric-compatible
connection, i.e., one for which $Q_{\mu\nu} = 0$ and $Q_\alpha{}^{\mu\nu} = 0$,
must feature torsion and, conversely, any torsion-free connection must
feature non-metricities with respect to the Galilean structure.
Additionally, even when $\dd\tau = 0$, differently to the Lorentzian case
there is \emph{not} a unique compatible torsion-free connection due to the
freedom in the choice of Coriolis field.  This is the main difference
to the Lorentzian case.  However, in the case $\dd\tau = 0$, given a
choice of observer $u^\mu$ and of Coriolis field $\ukappa{}_{\mu\nu}$, there
is a unique compatible torsion-free connection, given by\looseness-1
\begin{equation}
  \label{eq:uni_Gal_con_d_tau_0}
  \Gamma^\alpha_{\mu\nu}
  = \ucheckGamma{}^\alpha_{\mu\nu} + 2 \tau_{(\mu} \ukappa{}_{\nu)}{}^\alpha \; .
\end{equation}
Such torsion-free, metric-compatible connections are considered in
classical Newton--Cartan gravity (see, e.g., refs.~\cite{Kuenzle:1976, Schwartz:2026}).
However, once $\dd\tau \ne 0$, in addition to a Coriolis field, one needs
a choice of torsion and non-metricities compatible with \eqref{eq:Gal_Q_T} to fix a
connection.  In the next section, we will discuss a condition that
fixes a unique connection even in this case.

\begin{remark}
  If the clock form is foliation forming, which we assume in the
  present paper, then there is always a unique \emph{spatial} connection
  defined from the Galilean structure alone, namely the Levi-Civita
  connection of the induced Riemannian spatial metrics $h_{ij}$ on the
  spatial leaves $\Sigma$.
\end{remark}

\subsection{A unique boost-invariant connection}
\label{sec:unique_A_con}

The perspective on affine connections discussed in the previous
section was a `classification' one: \emph{starting} with a connection $\nabla$, it
determines---and is determined by---its torsion, non-metricities with
respect to the Galilean structure, and Coriolis field with respect to
a choice of timelike vector field $u^\mu$, according to \eqref{eq:con_general}.

\medskip

Here we now take a complementary pespective: we start with the
Galilean structure, a unit timelike vector field $u^\mu$, and a choice
of Coriolis field, and want to \emph{define} a connection from these
ingredients---and possibly some additional structures---by some
specified prescription.  In this approach, one may in general choose
the other ingredients to depend arbitrarily on $u^\mu$, and hence the
connections arising for different choices of $u^\mu$ could in general
differ.

However, in discussions of physical theories based on Galilean
geometry, one usually requires that the connection be \emph{independent} of
the choice of timelike vector field, since otherwise some preferred
timelike vector field, or class of preferred timelike vector fields,
could be defined (e.g., by requiring the connection to take some
simple form).  One demands that the theory should not possess a class
of preferred observers, i.e., preferred timelike vector fields, from
the geometric structure of the theory alone; any class of preferred
observers should appear only once a field equation is
considered.\footnote{For example, this is the case in standard
  Newton--Cartan gravity, for which the Newton--Cartan field equation
  induces a preferred class of observers called \emph{Galilean observers}
  (see \Cref{def:Gal_obs}).}  Therefore, because for our definition of Galilean Weyl
tensors we want to stay at a geometric, `off-shell' level, in the
present paper we need to consider connections independent of the
chosen timelike vector field.

In other words, we want to consider connections defined from $(\tau_\mu,
h^{\mu\nu})$, a timelike vector field $u^\mu$, and a Coriolis field in such
a way that they are invariant under the transformation $u^\mu \rightarrow u^\mu -
v^\mu$, where $v^\mu$ is a spacelike vector field.  Such a transformation
is called a \emph{(local) Galilean boost} with parameter $v^\mu$, and hence we
call connections defined by such a prescription \emph{boost-invariant
  connections}.  Note that for a boost-invariant
connection, the Coriolis field necessarily needs to transform under
boosts according to \eqref{eq:fix_con_Coriolis_trafo}: the connection is independent of the choice
of $u^\mu$, i.e., fixed, and hence the discussion from the previous
section applies.  The possibly remaining freedom in how to define a
boost-invariant connection resides in the choice of torsion and
non-metricities---which must also be boost-invariant (see \Cref{prop:boost-inv} in \Cref{app:proof_unique_A_con}
for details).

When $\dd\tau = 0$, there is a natural choice of torsion and
non-metricities that yields, for a given choice of $(u^\mu,
\ukappa{u}_{\mu\nu})$, a boost-invariant connection: we may choose the
torsion and non-metricities to vanish, yielding the connection \eqref{eq:uni_Gal_con_d_tau_0}.
In this case, the transformation law \eqref{eq:fix_con_Coriolis_trafo} for the Coriolis field
simplifies to
\begin{equation}
  \ukappa{}_{\mu\nu}
  \xrightarrow{u \to u - v} \ukappa{}_{\mu\nu}
  - \partial_{[\mu} \left(v_{\nu]} + \tau_{\nu]} \tfrac{v^2}{2} \right).
\end{equation}

Interestingly, as we prove in \Cref{app:proof_unique_A_con}, there is also a uniqueness result
for a boost-invariant connection when we only require $\tau \wedge \dd\tau = 0$,
similarly to the $\dd\tau=0$ case:
\begin{proposition}
  \label{prop:unique_A_con}
  Given a Galilean structure $(\tau_\mu, h^{\mu\nu})$ with $\tau\wedge\dd\tau = 0$ and a
  choice of timelike vector field and Coriolis field $(u^\mu,
  \ukappa{}_{\mu\nu})$, there is a unique boost-invariant connection,
  denoted $\unicon{\nabla}$,
  \begin{enumerate}
  \item whose Coriolis field with respect to $u^\mu$ is the given
    $\ukappa{}_{\mu\nu}$, and
  \item whose torsion and non-metricities are defined locally from $(\tau_\mu,
    h^{\mu\nu})$ and $u^\mu$, depending on derivatives of $\tau_\mu$, $h^{\mu\nu}$
    only to first order and linearly.
  \end{enumerate}

  It is given by
  \begin{equation} \label{eq:unique_A_con}
    \unicon{\Gamma}^\alpha_{\mu\nu}
    = \ucheckGamma{}^\alpha_{\mu\nu} + 2\tau_{(\mu} \ukappa{}_{\nu)}{}^\alpha
      - \frac{1}{2} \ub{}_{\mu\nu} \ator^\alpha \; ,
  \end{equation}
  or, equivalently, by
  \begin{equation}
    \unicon{Q}_{\mu\nu} = \omega_{\mu\nu} \; , \quad
    \unicon{Q}_\alpha{}^{\mu\nu} = -\delta^{(\mu}_\alpha \ator^{\nu)\vphantom{\mu}}_{\vphantom{\alpha}} \; , \quad
    \unicon{T}^\alpha{}_{\mu\nu} = 0 \; .
  \end{equation}

  Under boosts, the Coriolis field transforms as
  \begin{equation}
    \ukappa{}_{\mu\nu}
    \xrightarrow{u \to u - v}
    \ukappa{}_{\mu\nu}
    - \partial_{[\mu} \left(v_{\nu]} + \tau_{\nu]} \tfrac{v^2}{2} \right).
  \end{equation}
\end{proposition}

\begin{remark}
  When $\dd\tau=0$, the connection $\unicon{\nabla}$ reduces to the classical
  connection~\eqref{eq:uni_Gal_con_d_tau_0}.
\end{remark}

To our knowledge, this is the first time this connection is derived,
as boost-invariant connections considered in the literature are always
constructed with extra structures on top of $(\tau_\mu, h^{\mu\nu}, u^\mu,
\ukappa{}_{\mu\nu})$, such as a mass gauge field (see ref.~\cite{Hansen.Hartong.Obers:2020} and
references therein).  We include a presentation and discussion of
these approaches in \Cref{app:boost_inv_extra}.

The uniqueness of the connection~\eqref{eq:unique_A_con}, along with how its Coriolis
field transforms under Galilean boosts, is reminiscent of the natural
connection~\eqref{eq:uni_Gal_con_d_tau_0} obtained from a closed clock form.  In this sense, the
connection~\eqref{eq:unique_A_con} is the most natural connection that should be
considered from $(\tau_\mu, h^{\mu\nu}, u^\mu, \ukappa{u}_{\mu\nu})$ when $\tau\wedge\dd\tau =
0$.

\section{The Galilean limit}
\label{sec:Galilean_limit}

The definition of the Galilean Weyl tensor which we will provide in \Cref{sec:Gal_Weyl_conf}
will be obtained from the Galilean limit of the Lorentzian Weyl
tensor.  We detail that limit in the present section.  In particular,
we provide formulae for a non-closed clock form, and we provide the
limit of the 1+3 quantities.  As already remarked in \cref{sec:definitions}, we use the
terminology `Galilean' rather than `Newtonian' for denoting the $c \to
\infty$ limit of Lorentzian geometry to highlight the fact that this is a
geometric (off-shell) approach (see also our discussion in \Cref{sec:Newtonianity} on
the difference between a `Galilean theory' and a `Newtonian theory').

\subsection{The ansatz}
\label{sec:ansatz_limit}

In the classical ansatz for the Galilean limit, the first two leading
orders of the Lorentzian metric have the form
\begin{subequations} \label{eq:NR_limit}
\begin{align}
  \tayll{g}^{\mu\nu}
  &= h^{\mu\nu} + \lambda \, \left(-\ulim^\mu \ulim^\mu + k^{\mu\nu}\right) + \bigOl{2}
    , \\
  \tayll{g}_{\mu\nu}
  &= -\tfrac{1}{\lambda}\tau_\mu \tau_\nu + \blim_{\mu\nu} - 2\phi \tau_\mu \tau_\nu + \bigOl{1} ,
\end{align}
\end{subequations}
where $\ulim^\mu$ is a timelike vector field for which we write
$\blim_{\mu\nu} := \ub{\ulim}_{\mu\nu}$, $\phi$ is a scalar field, and $k^{\mu\nu}$
is a spacelike tensor.

\begin{remark}
  \label{rem:limit_gauge}
  In the geometry arising from the Galilean limit, there is a gauge
  freedom corresponding to the application of $\cc$-dependent
  diffeomorphisms to the Lorentzian objects \cite{Malik.Matravers:2008, Hansen.Hartong.Obers:2020, Hartong.Obers.Oling:2023}.  In particular, this
  includes the gauge transformation $\ulim^\mu \rightarrow \ulim^\mu - \psi \ator^\mu -
  h^{\mu\alpha} \partial_\alpha \psi$, where $\psi$ is a scalar field.  Interestingly, this
  gauge freedom leaves unchanged the vorticity tensor of any observer
  (see \Cref{sec:Gal_kin_lim}, in particular \Cref{rem:vort_free_observers}).
\end{remark}

For a general Galilean limit, the 1-form $\tau_\mu$ is not necessarily
exact, nor foliation forming.  In the literature, it is however almost
always assumed to be foliation forming, i.e., to satisfy $\tau \wedge \dd \tau =
0$.  This can be justified from the limit of the Einstein equation
with some assumptions on the leading order of the energy-momentum
tensor \cite{Dautcourt:1990b, Van_den_Bleeken:2017, Hansen.Hartong.Obers:2020}.  A stronger result was obtained in
ref.~\cite{Vigneron.Barzegar.Read:2025}, where the authors showed that, prior to considering any
field equations, the Frobenius integrability condition is necessarily
satisfied if global hyperbolicity is assumed for the Lorentzian
structure, and that the relation $\tau = \e^\psi \dd t$ holds globally,
where $\psi$ and $t$ are scalar fields.  We assume $\tau \wedge \dd\tau =
0$ for the rest of the paper.

\begin{remark}
  \textcite{Ergen.Hamamci.Van_den_Bleeken:2020} argued that choosing $\tau = \dd t$ is a gauge choice, without
  loss of generality for the limit, corresponding to a choice of
  foliation with respect to which performing the limit.  However, as
  we will be interested in conformal transformations of the Galilean
  Weyl tensor, we cannot fix the clock form to be exact.  Furthermore,
  one can argue that even though the choice of foliation might be seen
  as a choice of gauge on the Lorentzian level, performing the limit
  with respect to different foliations corresponds to taking the point
  of view of different families of observers, which are (in the limit)
  not related by Galilean boost transformations---i.e., when
  `forgetting' the Lorentzian origin of the Galilean theory, the
  choice of foliation becomes a \emph{physical} one.
\end{remark}

The spacetime volume form $\tayll{\epsilon}_{\alpha\beta\mu\nu}$ of the metric
$\tayll{g}_{\mu\nu}$ becomes
\begin{equation}
  \tayll{\epsilon}_{\alpha\beta\mu\nu}
  := \sqrt{\cc \, \left|\det\bigl(\tayll{g}_{\sigma\omega}\bigr)\right|} \, \varepsilon_{\alpha\beta\mu\nu} \\
  = \eta_{\alpha\beta\mu\nu} + \bigOl{1} \, .
\end{equation}

The \acronym{LO} term of the Levi-Civita connection is $\tayl{-1}{\Gamma}^\alpha_{\mu\nu}\;$,
and we have
\begin{align}
  \tayl{-1}{\Gamma}^\alpha_{\mu\nu}
  &= -\tau_\mu \tau_\nu \, \ator^\alpha \, ,\\
  \label{eq:Gamma_0}
  \tayl{0}{\Gamma}^\alpha_{\mu\nu}
  &= \ucheckGamma{\ulim}^\alpha_{\mu\nu}
    + \tau_\mu \tau_\nu \left(h^{\alpha\sigma} \partial_\sigma \phi - 2\phi \ator^\alpha - k^{\alpha \sigma} \ator_\sigma\right)
    - \ulim^\alpha \tau_{(\mu} \atorlim_{\nu)} \; ,
\end{align}
where we write $\atorlim_\mu := \blim_{\mu\alpha} \ator^\alpha$.

We see that the zeroth order term of the Levi-Civita connection is a
connection of the form~\eqref{eq:con_Gal_ator} with $c_1 = c_2 = c_4 = c_5 = 0$, $c_3 =
-1$, $d_i = 0$, and where the Coriolis field with respect to $\ulim^\mu$
is $\kappalim_{\mu\nu} := \ukappa{\ulim}_{\mu\nu} = \tau_{[\mu}
\left(\blim^{\alpha}{}_{\nu]} \partial_\alpha \phi - 2\phi \hat\ator_{\nu]} - \blim_{\nu]\alpha}k^{\alpha\beta}
  \hat\ator_\beta\right)$.  Therefore, this connection is not
boost-invariant, as it is different from the unique boost-invariant
connection~\eqref{eq:unique_A_con} (with the same Coriolis field $\kappalim_{\mu\nu}$).
However, when the clock form is closed, then $\tayl{-1}{\Gamma}^\alpha_{\mu\nu} =
0$, and $\tayl{0}{\Gamma}^\alpha_{\mu\nu}$ corresponds to the unique torsion-free
compatible connection, with $\kappalim_{\mu\nu} = \tau_{[\mu} \blim^{\alpha}{}_{\nu]}
\partial_\alpha \phi$.

The zeroth order connection is related to the boost-invariant
connection $\unicon\nabla$ (with the same Coriolis field $\kappalim_{\mu\nu}$)
by
\begin{equation}
  \label{eq:Gamma_0_to_uni}
  \tayl{0}{\Gamma}^\alpha_{\mu\nu}
  = \unicon\Gamma^\alpha_{\mu\nu} + \frac{1}{2} \blim_{\mu\nu} \ator^\alpha
    - \tau_{(\mu} \atorlim_{\nu)} \ulim^\alpha \; .
\end{equation}
Note that the spatial projection of the Coriolis field vanishes, i.e.,
we have
\begin{equation}
  \label{eq:Coriolis_lim}
  \kappalim^{\mu\nu} = h^{\alpha[\mu} \blim^{\nu]}{}_{\beta} \unicon\nabla_{\alpha} \ulim^\beta = 0 \, .
\end{equation}

The \acronym{LO} term of the Riemann tensor is $\tayl{-1}{R}^\alpha{}_{\beta\mu\nu}\,$, and
we have
\begin{subequations}
\begin{align}
  \tayl{-1}{R}^\alpha{}_{\beta\mu\nu}
    &= 2\tau_\beta \tau_{[\mu} \left(\unicon\nabla_{\nu]} \ator^\alpha - \tfrac{1}{2} \atorlim_{\nu]} \ator^\alpha \right) , \\
  \tayl{0}{R}^\alpha{}_{\beta\mu\nu}
    &= \iR{\tayl{0}{\nabla}}{R}^\alpha{}_{\beta\mu\nu}
      + \tau_\beta \tau_{[\mu} \left(\ulim^\alpha \ator^\sigma \Lie{\ulim} \blim_{\nu]\sigma}
      - 2 \ator^\alpha \partial_{\nu]} \phi \right)
      - \ator^\alpha \tau_{[\mu} \Lie{\ulim} \blim_{\nu]\beta} \; ,
\end{align}
\end{subequations}
and for the Ricci tensor
\begin{subequations}
\begin{align}
  \tayl{-1}{R}_{\mu\nu}
    &= - \tau_\mu \tau_\nu \unicon\nabla_\alpha \ator^\alpha \, , \\
%
  \tayl{0}{R}_{\mu\nu}
    &= \iR{\tayl{0}{\nabla}}{R}_{\mu\nu}
      + \tau_\mu \tau_\nu \, \ulim^\sigma \partial_\sigma \phi
      + \tau_{(\mu} \Lie{\ulim} \blim_{\nu)\sigma} \ator^\sigma \, ,
\end{align}
\end{subequations}
where $\iR{\tayl{0}{\nabla}}{R}^\alpha{}_{\beta\mu\nu}$ is the Riemann tensor of the
connection \eqref{eq:Gamma_0} arising as the zeroth order term of the Levi-Civita
connection.  Note that we used the unique boost-invariant connection
$\unicon\nabla$ in these expressions.

\subsection{Limit of the kinematical quantities}
\label{sec:limit_kin_vars}

In this section, we describe the Galilean limit of the kinematical
quantities.  We will use these later when expressing the electric and
magnetic Galilean Weyl tensors spatially.

\subsubsection{Limit of $U$-orthogonal quantities}
\label{sec:limit_spatial}

Any vector field $\tayll{U}^\mu$ timelike with respect to the Lorentzian
metric has an expansion of the form $\tayll{U}^\mu = u^{\mu} + \bigOl{1}$
where $u^\mu$ is unit timelike with respect to the Galilean structure
obtained in the limit.  Furthermore, for the associated one-form we
have $\tayll{U}_\mu = - \tfrac{1}{\cc} \tau_\mu + \bigOl{0}$, and
\begin{equation}
  \ubg{\tayll{U}}_{\mu\nu} = \ub{u}_{\mu\nu} + \bigOl{1} , \quad
  \ubg{\tayll{U}}^{\mu\nu} = h^{\mu\nu} + \bigOl{1} , \quad
  \ubg{\tayll{U}}^{\mu}{}_{\nu} = \ub{u}^{\mu}{}_{\nu} + \bigOl{1} .
\end{equation}
Therefore, as for the timelike vector field, the Lorentzian orthogonal
projector at \acronym{LO} corresponds to the Galilean covariant spatial metric for
the (0,2) version, to the Galilean spatial projector for the (1,1)
version, and to the spatial metric of the Galilean structure for the
(2,0) version.  There are two main consequences of this fact:
\begin{enumerate}
\item For any timelike vector field $\tayll{U}^\mu$, any contravariant
  tensor (say, $p^\mu$) orthogonal to $\tayll{U}^\mu$ will be annihilated
  by $\tau_\mu$ at \acronym{LO}, and therefore $\tayl{0}{p}^\mu$ defines, after
  pull-back, a spacelike tensor on $\Sigma$ (see \Cref{rem:proj}).  Using a basis
  adapted to the spatial foliation, we can denote that tensor with
  spatial indices as $\tayl{0}{p}^{\,i}$.  Therefore, while before
  taking the limit, a tensor orthogonal to $\tayll{U}^\mu$ cannot be
  defined on a `spatial foliation' if $\iR{\tayll{U}\,}{\Omega}_{\mu\nu} \ne 0$
  (since no such foliation exists), at \acronym{LO} that tensor always defines
  uniquely a spacelike tensor on the Galilean spatial foliation.
\item At \acronym{LO}, the spatial `connection' $\iR{\tayll{U}}{D}_\mu$ with respect
  to $\iR{\tayll{U}}{\ubg{}}_{\mu\nu}$ (defined in~\eqref{eq:def_D_spatial}) corresponds, after
  pull-back to $\Sigma$, to the unique spatial connection (denoted $D_i$)
  associated to the Galilean structure, namely the Levi-Civita
  connection of the Riemannian metric induced on $\Sigma$ by~$h^{\mu\nu}$.  The
  same holds for the \acronym{LO} of the `Riemann tensor' of that `connection',
  whose pull-back to $\Sigma$ we denote by $\CR_{ij}$.  Therefore, at \acronym{LO},
  $\iR{\tayll{U}}{D}_\mu$ has all the properties of an affine connection
  on $\Sigma$, even if $\iR{\tayll{U}\, }{\Omega}_{\mu\nu} \ne 0$ (see \Cref{rem:D_spatial}).
\end{enumerate}

\subsubsection{Galilean kinematical quantities}
\label{sec:Gal_kin_lim}

Given a Lorentzian timelike vector field $\tayll{U}^\mu$ with \acronym{LO} term
$u^\mu$, the \acronym{LO} terms of the kinematical quantities are
\begin{subequations} \label{eq:exp_kin}
\begin{align}
  \uTheta{\tayll{U}}^{\mu\nu}
  &= h^{\alpha\mu} h^{\beta\nu}\, \frac{1}{2} \Lie{u} \ub{u}_{\alpha\beta} + \bigOl{1} , \\
  \uOmega{\tayll{U}}^{\mu\nu}
  &= h^{\alpha\mu} h^{\beta\nu} \, \partial_{[\alpha} w_{\beta]}
    + \ator^{[\mu} w^{\nu]} + \bigOl{1} , \\
  \ua{\tayll{U}}^{\mu}
  &= - \frac{1}{\cc} \ator^\mu + \bigOl{1} ,
\end{align}
\end{subequations}
where $w^\mu := u^\mu - \ulim^\mu$ is the spacelike `tilt' / boost vector
field between the considered timelike vector field $u^\mu$ and the one
defined by the Galilean limit \eqref{eq:NR_limit}, and the index on $w_\beta$ was lowered
with $\ub{u}_{\mu\nu}$.  We want to define the Galilean expansion and
vorticity tensors of $u^\mu$ as these \acronym{LO} fields.  For the
\emph{(Galilean) expansion tensor of $u^\mu$}, this yields the definition
\begin{equation}
  \label{eq:lim_exp}
  \uTheta{u}^{\mu\nu}
  := h^{\alpha\mu} h^{\beta\nu}\, \frac{1}{2} \Lie{u} \ub{}_{\alpha\beta} \; .
\end{equation}
Note that due to $\ub{}_{\alpha\beta} u^\alpha = 0$, the Lie derivative $\Lie{u}
\ub{}_{\alpha\beta}$ vanishes when contracted with $u^\alpha$ (or $u^\beta$), such that
the index-lowered expansion tensor takes the simple form
$\uTheta{u}_{\mu\nu} = \frac{1}{2} \Lie{u} \ub{}_{\mu\nu}$.  For the
\emph{(Galilean) vorticity tensor of $u^\mu$}, we obtain the definition
\begin{equation}
  \label{eq:lim_vort}
  \uOmega{u}^{\mu\nu}
  := h^{\alpha\mu} h^{\beta\nu} \, \partial_{[\alpha} w_{\beta]} + \ator^{[\mu} w^{\nu]} \; .
\end{equation}
Note that this expression depends not only on $u^\mu$ and the Galilean
structure $(\tau_\mu, h^{\mu\nu})$, but also on the timelike vector field
$\ulim^\mu$ arising in the Galilean limit, since $w^\mu = u^\mu - \ulim^\mu$.
What are we to make of this dependence?  By definition, the Galilean
vorticity of $\ulim^\mu$ vanishes, i.e., $\uOmega{\ulim}^{\mu\nu} = 0$.
Further, if $\check u^\mu$ is \emph{any} timelike vector field with vanishing
Galilean vorticity, we see directly that the vorticity of $u^\mu$ can be
written as
\begin{equation}
  \label{eq:lim_vort_check_w}
  \uOmega{u}^{\mu\nu}
  = h^{\alpha\mu} h^{\beta\nu} \, \partial_{[\alpha} \check w_{\beta]} + \ator^{[\mu} \check w^{\nu]}
\end{equation}
with $\check w^\mu = u^\mu - \check u^\mu$.  In this sense, the notion of
Galilean vorticity depends, in addition to the Galilean structure, on
the \emph{family of vorticity-free observers} (timelike vector fields), which
is fixed by the Galilean limit, as it depends on the tilt $\check w_\mu$
with respect to any of these observers.  Note that this dependence is
already present in the classical case $\dd\tau = 0$.

By rewriting the first term of \eqref{eq:lim_vort_check_w} in terms of the boost-invariant
connection $\unicon\nabla$ arising in the Galilean limit, the dependence on
the family of vorticity-free observers can in a certain sense be
completely `pushed' into the second, $\ator^\mu$-dependent term, as
follows. By a direct computation, one can show
\begin{equation}
  \label{eq:Coriolis_tilde_nabla}
  h^{\alpha\mu} h^{\beta\nu} \, \partial_{[\alpha} \check w_{\beta]}
  = h^{\alpha[\mu} \ub{}^{\nu]}{}_\beta \unicon\nabla_\alpha u^\beta
  - h^{\alpha[\mu} \, \check{\ub{}}^{\nu]}{}_\beta \unicon\nabla_\alpha \check u^\beta \; ,
\end{equation}
where we wrote $\check{\ub{}}^\mu{}_\nu := \ub{\check u}^\mu{}_\nu$.  This is
precisely the difference of the spatial projections of the Coriolis
fields of $\unicon\nabla$ with respect to $u^\mu$ and $\check u^\mu$.  Thus,
the vorticity can be written as
\begin{equation}
  \uOmega{u}^{\mu\nu}
  = \ukappa{u}^{\mu\nu} - \ukappa{\check u}^{\mu\nu}
    + \ator^{[\mu} \check w^{\nu]} \; ,
\end{equation}
where, as above, $\check u^\mu$ is a member of the family of
vorticity-free timelike vector fields.  In particular, in terms of the
field $\ulim^\mu$ arising in the Galilean limit, by \eqref{eq:Coriolis_lim} we have
\begin{equation}
  \uOmega{u}^{\mu\nu}
  = \ukappa{u}^{\mu\nu} + \ator^{[\mu} w^{\nu]} \; .
\end{equation}
In the classical case $\dd\tau = 0$, the vorticity of $u^\mu$ can hence be
rewritten in a form that does not depend \emph{explicitly} on the family of
vorticity-free observers, at the expense of expressing it in terms of
the spacetime connection.  Nevertheless, even in this case
$\uOmega{u}_{ij}$ is still imposed to be an exact (spatial) 2-form
by~\eqref{eq:Coriolis_tilde_nabla}, which is related, in particular, to the so-called \emph{Newtonian
  condition} introduced by Trautman \cite{Trautman:1963, Trautman:1965}.

Note that the Galilean expansion and vorticity tensors as defined here
reduce to those from standard Newton--Cartan gravity when assuming
$\dd\tau = 0$ (see \cite[Section~3.2 and Remark~3.51]{Schwartz:2026}).

\begin{remark}
  \label{rem:vort_free_observers}
  In the classical Galilean limit, i.e., with $\dd\tau = 0$, the family
  of vorticity-free observers is easily characterised: any two of them
  differ by a spacelike vector field whose associated spatial one-form
  is (spatially) closed.  In what follows, we generalise this result
  in the presence of a non-trivial time dilation vector~field.

  Taking the exterior derivative of the equation $\dd\tau = \tau \wedge \ator$,
  we obtain $0 = -\tau \wedge \dd\ator$, implying that the spatial pull-back
  $\ator_i$ of the time dilation vector field to any spatial leaf $\Sigma$
  is a closed (spatial) 1-form.  In particular, locally we have
  $\ator_i = -\partial_i \Psi$ for some scalar field $\Psi$, i.e.,
  \begin{equation}
    \label{eq:ator_local}
    \ator\big|_\Sigma = -\dd\bigl(\Psi\big|_\Sigma\bigr) \, .
  \end{equation}

  Let now $\check w^\mu$ be the difference between an arbitrary timelike
  vector field $u^\mu$ and a vorticity-free timelike vector field
  $\check u^\mu$.  From \eqref{eq:lim_vort_check_w}, we know that $u^\mu$ is itself vorticity-free
  if and only if the spacelike form $\check w_i$ satisfies
  \begin{equation}
    0 = \dd\bigl(\check w|_\Sigma\bigr) + \ator \wedge \check w \, .
  \end{equation}
  Combined with \eqref{eq:ator_local}, this is equivalent to
  \begin{equation}
    0
    = \dd\bigl(\check w|_\Sigma\bigr)
      - \dd\bigl(\Psi\big|_\Sigma\bigr) \wedge \check w
    = \e^\Psi \dd\bigl(\e^{-\Psi} \check w|_\Sigma\bigr) \, ,
  \end{equation}
  i.e., to
  \begin{equation}
    \label{eq:vorticity_free_condition}
    \check w_i = \e^\Psi n_i \; ,
  \end{equation}
  where $n_i$ is a spatially closed 1-form, $\dd\bigl(n|_\Sigma\bigr) = 0$.
  \pagebreak

  In other words, the Galilean limit implies the existence of a family
  of (Galilean) vorticity-free observers that all differ by a spatial
  velocity that is locally of the form~\eqref{eq:vorticity_free_condition}. This form includes in
  particular the gauge freedom on $\hat u^\mu$ discussed in \Cref{rem:limit_gauge}, for
  which $n_i$ is exact.

  As mentioned in \Cref{sec:ansatz_limit}, global hyperbolicity of the Lorentzian
  structure imposes $\tau = \e^\Psi \dd t$ to hold globally, such that \eqref{eq:ator_local}
  holds globally, and the difference of vorticity-free observers has
  the form~\eqref{eq:vorticity_free_condition} globally.
\end{remark}

The Galilean expansion and vorticity tensors \eqref{eq:lim_exp}, \eqref{eq:lim_vort} are spacelike
with respect to the Galilean structure, and therefore define tensors
on each spatial leaf $\Sigma$ (see \Cref{sec:limit_spatial}).  We denote these spatial tensors
by ${\uTheta{u}}_{ij}$ and ${\uOmega{u}}_{ij}$, respectively, and call
them by the names of \emph{(Galilean) expansion} and \emph{vorticity} tensors of
$u^\mu$ as well.  We will discuss some of their properties in \Cref{sec:inter_mag}.

Let us stress again that, while the \acronym{LO} term of the spatial connection
$\iR{\tayll{U}}{D}_\mu$ does not depend on the choice of timelike vector
field, the \acronym{LO} terms of the expansion and vorticity tensors \emph{do} depend
on that vector field---they are properties of the vector field, not of
the spacetime.  For two different Galilean timelike vector fields
$u^\mu$ and $m^\mu$, differing by the spacelike vector field $v^\mu := u^\mu -
m^\mu$, from \eqref{eq:lim_exp} and \eqref{eq:lim_vort}, their Galilean expansion and vorticity tensors
differ by
\begin{equation}
  \label{eq:diff_kin}
  \uTheta{m}_{ij} - \uTheta{u}_{ij}
  = -D_{(i} v_{j)} \; , \quad
  \uOmega{m}_{ij} - \uOmega{u}_{ij}
  = -\partial_{[i} v_{j]} - \ator_{[i} v_{j]} \; .
\end{equation}

Finally, to later assess the conformal invariance of the Galilean Weyl
tensor, we need to determine how the Galilean kinematical quantities
change under conformal transformation.  From~\eqref{eq:conf_Lor_kin_vars} for the
transformation of the Lorentzian kinematical quantities, as well as
the transformation of the spatial connection and spatial Ricci tensor,
we obtain
\begin{subequations} \label{eq:conf_Gal_kin_vars}
\begin{align}
  \theta
  &\overset{\varphi}{\rightarrow} \e^{-\varphi} \left(\theta  + 3\,  \Lie{u} \varphi \right) \; , \quad
  \sigma_{ij}
    \overset{\varphi}{\rightarrow} \e^{\varphi} \sigma_{ij} \; , \quad
  \Omega_{ij}
    \overset{\varphi}{\rightarrow} \e^{\varphi} \Omega_{ij} \; , \quad
  \ator_{i}
    \overset{\varphi}{\rightarrow} \ator_{i} - D_i \varphi \; ,  \\
  \iR{h\,}{\Gamma}^k_{ij}
  &\overset{\varphi}{\rightarrow} \iR{h\,}{\Gamma}^k_{ij} + 2 h^k{}_{(i} D_{j)} \varphi
    - h_{ij} D^k \varphi \; , \\
  \CR_{ij}
  &\overset{\varphi}{\rightarrow} \CR_{ij} - D_i D_j \varphi + D_i \varphi D_j \varphi
    - \left(\Delta \varphi + D_k \varphi D^k \varphi \right) h_{ij} \; ,
\end{align}
\end{subequations}
where $\iR{h\,}{\Gamma}^k_{ij}$ are the Levi-Civita coefficients of the
spatial metric.

\section{The Galilean Weyl tensor}
\label{sec:Gal_Weyl_conf}

In this section we show how to define a Galilean Weyl tensor that is
manifestly conformally invariant.  We first provide a definition
involving a spacetime connection, then we provide a definition
involving the spatial connection and the kinematical quantities.  Both
approaches are off-shell.

\subsection{Spacetime approach}

Here, we will discuss how to give a definition of a Galilean Weyl
tensor using a spacetime connection by consideration of the Galilean
limit.

\subsubsection{Definition}

To obtain a definition of the Galilean Weyl tensors (i.e., versions of
`the' Galilean Weyl tensor with different index placement), we
consider the Galilean limit of the corresponding Lorentzian versions,
and keep only the \acronym{LO} terms.

For simplicity, in the spacetime approach, we consider only a timelike
vector field $\tayll{U}^\mu$ such that its \acronym{LO} term is $\ulim^\mu$, i.e.,
the timelike vector field naturally defined by the Galilean limit
\eqref{eq:NR_limit}.  In other words, we consider a timelike vector
field $\tayll{U}^\mu$ for which the Galilean vorticity vanishes (see
\Cref{sec:Gal_kin_lim}).  Only in the spatial approach of \Cref{sec:Weyl_spatial_approach} will a general timelike
vector field be considered.

The \acronym{LO} terms of the electric and magnetic parts of the Lorentzian Weyl
tensor with respect to to $\tayll{U}^\mu$, defined by~\eqref{eq:def_E} and \eqref{eq:def_H}, are
\begin{align}
  \label{eq:E_LO}
  \Elim_{\mu\nu}
  := \tayl{-1}{E}_{\mu\nu}
  &=\frac{1}{2} \blim^\alpha{}_{\langle\mu}  \blim^\beta{}_{\nu\rangle} \tayl{0}{R}_{\alpha\beta}
    - \blim^\alpha{}_{\langle\mu} \blim^\beta{}_{\nu\rangle} \unicon\nabla_\alpha \atorlim_\beta
    + \atorlim_{\langle\mu} \atorlim_{\nu\rangle} \;, \nonumber \\
  &=\frac{1}{2} \blim^\alpha{}_{\langle\mu} \blim^\beta{}_{\nu\rangle} \iR{\tayl{0}{\nabla}}{R}_{\alpha\beta}
    - \blim^\alpha{}_{\langle\mu} \blim^\beta{}_{\nu\rangle} \unicon\nabla_\alpha \atorlim_\beta
    + \atorlim_{\langle\mu} \atorlim_{\nu\rangle} \; , \\
  \label{eq:H_LO}
  \Hlim_{\mu\nu}
  := \tayl{0}{H}_{\mu\nu}
  &= \frac{1}{2} \ueta{}_{\alpha\beta(\mu} h^{\beta\lambda}\tayl{0}{R}^\alpha{}_{\lambda\nu)\sigma} {\ulim}^\sigma
    \nonumber \\
  &= \frac{1}{2} \ueta{}_{\alpha\beta(\mu} h^{\beta\lambda} \iR{\tayl{0}{\nabla}}{R}^\alpha{}_{\lambda\nu)\sigma} {\ulim}^\sigma
    + \frac{1}{4} \ator^\alpha \, h^{\sigma\beta} \eta_{\alpha\sigma}{}_{(\mu} \, \Lie{\ulim} \blim_{\nu)\beta}\; .
\end{align}

Using~\eqref{eq:Weyl_E_H}, we can express the \acronym{LO} term of the full Weyl tensor in terms
of the \acronym{LO} terms of the electric and magnetic parts.  Considering the
different `versions' of the full Weyl tensor, we obtain
\begin{subequations}
\begin{align}
  \label{eq:13C_LO}
  \tayll{C}_{\alpha\beta\mu\nu}
  &= \frac{1}{\cc} \left[-2\tau_\alpha \tau_{[\mu} {\Elim}_{\nu]\beta}
    + 2\tau_\beta \tau_{[\mu} {\Elim}_{\nu]\alpha} \right] + \bigOl{0} , \\
  \tayll{C}^\alpha{}_{\beta\mu\nu}
  &= \frac{1}{\cc}\left[-2\tau_\beta \tau_{[\mu} {\Elim}_{\nu]}{}^\alpha \right] + \bigOl{0} , \\
  \label{eq:22C_LO}
  \tayll{C}^{\alpha\beta}{}_{\mu\nu}
  &= 4\left(\blim^{[\alpha}{}_{[\mu} {\Elim}{}^{\beta]}{}_{\nu]} - \ulim^{[\alpha}\tau_{[\mu} {\Elim}{}^{\beta]}{}_{\nu]} \right)
    - 2 \eta^{\alpha\beta\sigma} \tau_{[\mu} {\Hlim}_{\nu]\sigma} + \bigOl{1} , \\
  \tayll{C}^{\alpha\beta\mu}{}_{\nu}
  &= 2\left(
    h^{\mu[\alpha} {\Elim}{}^{\beta]}{}_{\nu}\
    - \blim^{[\alpha}{}_\nu {\Elim}{}^{\beta]\mu}
    + \tau_\nu \ulim^{[\alpha} {\Elim}{}^{\beta]\mu}
    \right)
    - \tau_\nu \eta^{\alpha\beta\sigma} h^{\mu\omega} {\Hlim}_{\sigma\omega} + \bigOl{1} , \\
  \tayll{C}^{\alpha\beta\mu\nu}
  &= 2\left(h^{\alpha[\mu} {\Elim}{}^{\nu]\beta} - h^{\beta[\mu} {\Elim}{}^{\nu]\alpha} \right) + \bigOl{1} ,
\end{align}
\end{subequations}
where we recall that indices on the Galilean objects are raised with
$h^{\mu\nu}$.

We see that the \acronym{LO} terms of the (0,4), the (1,3), and the (4,0)
versions of the Weyl tensor depend only on the \acronym{LO} term of the electric
part, while the \acronym{LO} terms of the (2,2) and the (3,1) versions depend on
the \acronym{LO} terms of both the electric and magnetic parts.  In this sense,
the (2,2) and the (3,1) Weyl tensors carry more geometric information
than the other versions and for this reason they are more suited to
define a Galilean Weyl tensor.  This peculiarity does not occur for
the electric and magnetic parts themselves, for which the \acronym{LO} terms
carry the same information regardless of the limit being taken with
upper or lower indices.  In this sense, \emph{what carries all the
  information on the Weyl tensor at \acronym{LO} are its electric and magnetic
  parts}.

Note also that for the different versions of the Weyl tensor, the \acronym{LO}
is \emph{different}: the (0,4) and the (1,3) versions have a non-vanishing
$\frac{1}{\cc}$ coefficient, while the others start at order $\cc^0$.
(In this sense, the (0,4) and (1,3) Weyl tensors diverge in the
Galilean $c \to \infty$ limit.)

There is, however, an issue in the \acronym{LO} expressions given so far: the
Ricci and Riemann tensors involved in these formulae are not defined
with respect to a boost-invariant connection, but with respect to the
zeroth order connection~\eqref{eq:Gamma_0}.  It is therefore preferable to express
them in terms of the unique boost-invariant connection~\eqref{eq:unique_A_con}.  Using~\eqref{eq:Gamma_0_to_uni}
to compute $\iR{\tayl{0}{\nabla}}{R}^\alpha{}_{\beta\mu\nu} - \unicon R^\alpha{}_{\beta\mu\nu}$,
we obtain
\begin{align}
  \Elim_{\mu\nu}
    &=\frac{1}{2} \blim^\alpha{}_{\langle\mu}  \blim^\beta{}_{\nu\rangle} \unicon{R}_{\alpha\beta}
      - \blim^\alpha{}_{\langle\mu}  \blim^\beta{}_{\nu\rangle} \unicon\nabla_\alpha \atorlim_\beta
      + \frac{3}{4} \atorlim_{\langle\mu} \atorlim_{\nu\rangle} \;, \\
  \Hlim_{\mu\nu}
    &= \frac{1}{2} \ueta{}_{\alpha}{}^\beta{}_{(\mu} \unicon{R}^\alpha{}_{\beta\nu)\sigma} \ulim^\sigma
      + \frac{1}{4} \ator^\alpha\, \eta_{\alpha}{}^{\beta}{}_{(\mu} \, \Lie{\ulim} \blim_{\nu)\beta}\; .
\end{align}

We can now give a definition for the off-shell Galilean Weyl tensor,
its electric part, and its magnetic part.

\begin{definition}[Galilean Weyl tensor, spacetime version]
  \label{def:Gal_Weyl_off_shell_restricted}
  Let $(\tau_\mu, h^{\mu\nu})$ be a Galilean structure, $(u^\mu, \ukappa{}_{\mu\nu})$
  be a timelike vector field and a choice of spatially vanishing
  Coriolis field, i.e\ $\kappa^{\mu\nu} = 0$, and $\unicon\nabla$ be the unique
  boost-invariant connection~\eqref{eq:unique_A_con} constructed from $(\tau_\mu, h^{\mu\nu}, u^\mu,
  \ukappa{}_{\mu\nu})$.  The (off-shell) \emph{Galilean electric and
    magnetic Weyl tensor parts} with respect to $u^\mu$ are defined,
  respectively, as
  \begin{align}
    \label{eq:E_LO_uni}
    {E}_{\mu\nu}
    &:=\frac{1}{2} \ub{}^\alpha{}_{\langle\mu}  \ub{}^\beta{}_{\nu\rangle} \unicon{R}_{\alpha\beta}
      - \ub{}^\alpha{}_{\langle\mu}  \ub{}^\beta{}_{\nu\rangle} \unicon\nabla_\alpha \uator{}_\beta
      + \frac{3}{4} \uator{}_{\langle\mu} \uator{}_{\nu\rangle} \; , \\
    \label{eq:H_LO_uni}
    {H}_{\mu\nu}
    &:= \frac{1}{2} \ueta{}_{\alpha\beta(\mu} h^{\beta\lambda} \unicon{R}^\alpha{}_{\lambda\nu)\sigma} {u}^\sigma
      + \frac{1}{4} \uator{}_\alpha \eta^{\alpha\beta}{}_{(\mu} \, \Lie{u} \! \ub{}_{\nu)\beta}
      \; .
  \end{align}
  The \emph{(2,2) Galilean Weyl tensor} is defined as
  \begin{equation}
    \label{eq:22C_LO_uni}
    C^{\alpha\beta}{}_{\mu\nu}
    := 4\left(\ub{}^{[\alpha}{}_{[\mu} {E}^{\beta]}{}_{\nu]}
        - u^{[\alpha}\tau_{[\mu} {E}^{\beta]}{}_{\nu]} \right)
      - 2 \eta^{\alpha\beta\sigma} \tau_{[\mu} {H}_{\nu]\sigma} \; .
  \end{equation}
\end{definition}

\begin{proposition}
  The Galilean Weyl tensor ${C}^{\alpha\beta}{}_{\mu\nu}$, its electric part
  ${E}_{\mu\nu}$, and its magnetic part ${H}_{\mu\nu}$ are invariant under
  Galilean conformal transformations.
\end{proposition}

\begin{proof}
  Because the Lorentzian electric and magnetic parts
  $\uE{\tayll{U}}_{\mu\nu}$ and $\uH{\tayll{U}}_{\mu\nu}$ are invariant under
  Lorentzian conformal transformations, and because the Galilean limit
  of a Lorentzian conformal transformation is a Galilean conformal
  transformation, their \acronym{LO} terms are automatically invariant
  under Galilean conformal transformations.
\end{proof}

Note that when $\dd\tau = 0$, the definition~\eqref{eq:22C_LO_uni} reduces to
\begin{equation}
  \label{eq:Weyl_d_tau_0}
  C^{\alpha\beta}{}_{\mu\nu}
  = h^{\sigma\beta} R^\alpha{}_{\sigma\mu\nu}
  - 2 \delta^{[\alpha}{}_{[\mu} R_{\nu]\sigma} h^{\beta]\sigma}
  + \frac{1}{3} \delta^{\alpha}{}_{[\mu} \delta^\beta{}_{\nu]} R_{\sigma\omega}h^{\sigma\omega}\;.
\end{equation}

Because all information on the Galilean Weyl tensor is contained in
the electric and magnetic parts, and because there is no unique way to
define a `full Galilean Weyl tensor' from the electric and magnetic
parts, for the rest of this paper we will only discuss these two
parts, and not anymore consider the full object~\eqref{eq:22C_LO_uni}.

\subsubsection{Discussion}
\label{sec:discuss_spacetime_Weyl}

It is important to note that the precise formulae for the electric and
magnetic parts depend on which spacetime connection we choose to
express the Ricci and Riemann tensors.  We chose the connection
derived in \Cref{sec:unique_A_con} because we wanted a boost-invariant connection that
does not require the introduction of an additional structure.  As
discussed in \Cref{app:boost_inv_extra}, by introducing an extra structure, e.g., a mass
gauge field $M_\mu$, other boost-invariant connections could be
considered, changing essentially the last terms in the formulae~\eqref{eq:E_LO_uni}
and~\eqref{eq:H_LO_uni}, while changing nothing in the formula~\eqref{eq:22C_LO_uni} for the (2,2)
Galilean Weyl tensor.  In any case, regardless of the connection that
is chosen, the formulae will always be conformally invariant and
equivalent.

For instance, by introducing the mass gauge field $M_\mu$ and the
boost-invariant fields $\bar u^\mu$, $\ub{\bar u}_{\mu\nu}$, and
$\bar\ator_\mu$ (see \Cref{app:con_Mass} for the definitions), we can choose the
connection to be either \emph{torsionful and compatible}, taking the form
\begin{equation}
  \bar\Gamma^\alpha_{\mu\nu}
  = \ucheckGamma{\bar u}^\alpha_{\mu\nu} + 2\tau_{(\mu} \ukappa{\bar u}_{\nu)}{}^\alpha
    + \ub{\bar u}^\alpha{}_{(\mu} \bar \ator_{\nu)}
    - \ub{\bar u}_{\mu\nu} \bar \ator^\alpha
    + \ub{\bar u}^\alpha{}_{[\mu} \bar \ator_{\nu]}
    + \ub{\bar u}^\alpha \tau_{[\mu} \bar \ator_{\nu]} \; ,
\end{equation}
or \emph{torsion-free and conformally compatible} (see definition~\eqref{eq:conf_connection}),
taking the form
\begin{equation}
  \bar\Gamma^\alpha_{\mu\nu}
  = \ucheckGamma{\bar u}^\alpha_{\mu\nu} + 2\tau_{(\mu} \ukappa{\bar u}_{\nu)}{}^\alpha
    + 2 \ub{\bar u}^\alpha{}_{(\mu} \bar \ator_{\nu)}
    - \ub{\bar u}_{\mu\nu} \bar \ator^\alpha
    + \, \bar u^\alpha \tau_{(\mu} \bar \ator_{\nu)} \; .
\end{equation}
In these two cases, the electric and magnetic parts become simply
\begin{equation}
  {E}_{\mu\nu}
  = \frac{1}{2} \ub{\bar u}^\alpha{}_{\langle\mu} \ub{\bar u}^\beta{}_{\nu\rangle} \bar{R}_{\alpha\beta} \;, \quad
  {H}_{\mu\nu}
  = \frac{1}{2} \ueta{}_{\alpha\beta(\mu} h^{\beta\lambda} \bar{R}^\alpha{}_{\lambda\nu)\sigma} {\bar u}^\sigma \; .
\end{equation}
These formulae are equivalent to those from \Cref{def:Gal_Weyl_off_shell_restricted}, but written with
respect to another spacetime connection.  They have the merit of being
simpler, but at the cost of introducing extra structure.

Another clear disadvantage of \Cref{def:Gal_Weyl_off_shell_restricted} is that the objects which appear
therein are difficult to interpret, regardless of the spacetime
connection we choose.  In particular, it is not clear how the spatial
curvature can be obtained from the spacetime connection, unless $\dd\tau
= 0$.  For this reason, a more enlightening way of defining the
electric and magnetic parts of the Galilean Weyl tensor would be to
use only spacelike quantities, like the spatial curvature and
kinematical quantities.  The clear advantage of such a `spatial
approach' is that there is a unique spatial connection that can be
constructed from the Galilean structure, namely the Levi-Civita
connection of the (pull-back) Riemannian metric $h_{ij}$ on $\Sigma$.
Therefore, there will be a unique way of representing the electric and
magnetic parts.  We pursue this approach in the next section.

Finally, note that \Cref{def:Gal_Weyl_off_shell_restricted} defines the electric and magnetic parts with
respect to an irrotational observer, as imposed by the condition
$\kappa^{\mu\nu} = 0$.  In principle, it would be possible to obtain a
definition for a general observer by taking the limit with respect to
a timelike vector field $\tayll{U}$ with a \acronym{LO} term different from
$\ulim$, which would result in more complicated formulae.  We will
provide such a general definition only in the spatial approach in the
next section.

Note also that when lowering the second index of the $\dd\tau = 0$ form \eqref{eq:Weyl_d_tau_0}
of the Galilean Weyl tensor with $\ub{}_{\mu\nu}$, the resulting object
fulfills all the conditions demanded in ref.~\cite{Dewar.Weatherall:2018} of a Galilean Weyl
tensor, apart from spatial flatness.

\subsection{Spatial approach}
\label{sec:Weyl_spatial_approach}

Here, we will discuss a `purely spatial' definition of a Galilean Weyl
tensor.

\subsubsection{Definition}

From the expressions~\eqref{eq:E_kin_vars} and~\eqref{eq:H_kin_vars} of the electric and magnetic parts of
the Lorentzian Weyl tensor in terms of the kinematical quantities, and
the limit of kinematical quantities presented in \Cref{sec:limit_kin_vars}, we obtain the
following definition:

\begin{definition}[Galilean Weyl tensor, spatial version]
  \label{def:Gal_Weyl_off_shell_spatial}
  Let $(\tau_\mu, h^{\mu\nu})$ be a Galilean structure such that $\tau_\mu$ is
  foliation forming with leaves $\Sigma$ and with time dilation vector
  field $\ator^\mu$.  Given a timelike vector field $u^\mu$ with Galilean
  shear $\sigma_{\mu\nu}$ and vorticity $\Omega_{\mu\nu}$, the \emph{electric and magnetic
    parts of the Galilean Weyl tensor} on $\Sigma$ with respect to $u^\mu$ are
  defined as
  \begin{align}
    E_{ij}
    &:= \frac{1}{2} \left(\CR_{\langle ij \rangle}
      - D_{\langle i} \ator_{j\rangle}
      + \ator_{\langle i} \ator_{j\rangle} \right) , \\
    \label{eq:def_H_Gal}
    H_{ij}
    &:= \eta_{ab\langle i}\left[D^a \left(\sigma^b{}_{j\rangle} + \Omega^b{}_{j\rangle} \right)
      + \ator_{j\rangle} \Omega^{ab}\right] ,
  \end{align}
  where $D_i$ and $\CR_{ij}$ are, respectively, the Levi-Civita
  connection of the spatial metric $h_{ij}$ and its Ricci tensor.
\end{definition}

Note that in the spatial approach, we have only defined the electric
and magnetic parts of the Galilean Weyl tensor, not any `spatial
Galilean Weyl tensor' as a combined object.

By a direct computation, one can show:
\begin{proposition}
  The (spatial) electric part ${E}_{ij}$ and magnetic part ${H}_{ij}$
  of the Galilean Weyl tensor are invariant under Galilean conformal
  transformations~\eqref{eq:conf_Gal_kin_vars}.
\end{proposition}

Because $D_i$, $\CR_{ij}$ and $\ator_i$ are independent of the choice
of observer, the electric part is independent as well.  This is not
the case for the magnetic part.

\subsubsection{Relation to the spacetime approach}

The spatial \Cref{def:Gal_Weyl_off_shell_spatial} of the electric and magnetic parts of the Galilean
Weyl tensor is equivalent to the (pull-back of the) spacetime \Cref{def:Gal_Weyl_off_shell_restricted},
because they both correspond to the \acronym{LO} terms of the electric and
magnetic parts of the Lorentzian Weyl tensor.  However, while there
was the freedom of the choice of spacetime connection in \Cref{def:Gal_Weyl_off_shell_restricted}, there is
no such freedom in \Cref{def:Gal_Weyl_off_shell_spatial}.

To obtain the spacetime \Cref{def:Gal_Weyl_off_shell_restricted} from the spatial \Cref{def:Gal_Weyl_off_shell_spatial} without consideration
of the Galilean limit, one should express the spatial covariant
derivative and the spatial Ricci tensor in terms of the chosen
spacetime connection.  When $\dd\tau = 0$, that relation is similar to
the relativistic one~\eqref{eq:def_D_spatial} (see, e.g., \cite{Vigneron:2021}): the spatial Levi-Civita
connection is simply the restriction of the spacetime connection.
However, this is not the case when $\dd\tau\ne0$.  For instance,
restriction of the unique boost-invariant connection $\unicon\nabla_\mu$ does
not lead to a connection $D_i$ that is compatible with the spatial
Riemannian metric.  A more elaborate definition is required, which we
do not aim to study in this paper as we already established the
equivalence between \Cref{def:Gal_Weyl_off_shell_restricted} and \Cref{def:Gal_Weyl_off_shell_spatial}.

The fact that, when $\dd\tau \ne 0$, there arises an ambiguity in the
spacetime connection that should be considered in the Galilean limit,
and no such ambiguity arises for the spatial connection, may be taken
as an indication that, in the $\dd\tau \ne 0$ case, to understand the
geometry and the physics of a Galilean invariant theory it is easier
\emph{not} to work with a spacetime representation \emph{à la} standard
Newton--Cartan geometry, but to work instead with a spatial
representation, closer to the classical Newtonian formulation.

\section{Discussion}
\label{sec:disc}

The approach we followed in \Cref{sec:Gal_Weyl_conf} for the definition of a Galilean Weyl
tensor---or its electric and magnetic parts---is completely off-shell.
In particular, we did not impose the spatial curvature to be zero,
contrary to what has been done so far in the literature using Galilean
structures to define the Galilean Weyl tensor \cite{Ehlers.Buchert:2009, Dewar.Weatherall:2018, Dewar.Read:2020}, and
the literature using a more classical approach to the Newtonian limit
\cite{Ellis:1971, Ellis:2009, Kofman.Pogosyan:1995}.  This has two advantages.  First, since spatial
curvature is not preserved under conformal transformations, having
non-zero spatial curvature is the only way in which one can define a
conformally invariant Galilean Weyl tensor.  Second, the Weyl tensor
of Lorentzian structures is a geometric quantity defined independently
of any field equations.  It is therefore natural to require that a
Galilean version should also be defined geometrically.

Let us elaborate on the first point above.  One might wonder how it
can be, given what we have just said, that \textcite{Dewar.Read:2020} profess to show that
their \emph{on-shell} Gailean Weyl tensor is invariant under Galilean
conformal transformations.  The key to resolving this puzzle is to
note that the conformal rescalings chosen in ref.~\cite{Dewar.Read:2020} are \emph{not} the
same as the conformal rescalings used in this article (or in refs.~\cite{Schwartz.Read.Vigneron:2026,
  Wolf.Read.Vigneron:2024}): in fact, those scalings suffice to make the conformal
transformations considered in that article \emph{projective} transformations
(on which see, e.g., ref.~\cite{Malament:2012}).  And since the `Galilean Weyl
tensor' considered in refs.~\cite{Ehlers.Buchert:2009, Dewar.Weatherall:2018, Dewar.Read:2020} is in fact what is
known as the \emph{projective} Weyl tensor (see, e.g., ref.~\cite{Luebbe:2013}), which is
a \emph{projective} invariant, it is of course no surprise that that object
is invariant under those alternative transformations.

This, in turn, leads to a related puzzle: how can it be the case that,
by taking the Galilean limit of a relativistic Weyl tensor which is \emph{conformally}
invariant, refs.~\cite{Ehlers.Buchert:2009, Dewar.Weatherall:2018} obtain an object which is \emph{projectively}
invariant?  The answer to this question (setting aside the fact that
there are some mathematical slips in ref.~\cite{Ehlers.Buchert:2009} when it comes to
taking the limit) is that it is known from ref.~\cite{Luebbe:2013} that, in the
relativistic setting, the projective and conformal Weyl tensors agree
when one is working with an Einstein manifold (for which, by
definition, $g_{\mu\nu} \propto R_{\mu\nu}$).  But since when one goes on-shell in
Newton--Cartan theory one effectively has this condition (recall the
Newton--Cartan field equation $R_{\mu\nu} = 4\pi G \rho \tau_\mu \tau_\nu$), it is little
surprise that one obtains a projective invariant in the end.  Of
course, if one follows the approach of the present paper and remains
fully off-shell, what one arrives at for the Galilean Weyl tensor (as
in Definition \ref{eq:22C_LO_uni}) is \emph{not} a projective invariant (but \emph{is} invariant
under the `correct' Galilean conformal transformations\footnote{For a
  more general approach to Galilean conformal transformations which
  subsumes both choices of scaling, see ref.~\cite{March.Read:2026};
  this more general definition is also used in
  ref.~\cite{Adlam.Linnemann.Read:2025}.}).

In any case, let us return now to our off-shell definitions as
presented above.  From the discussion in the previous section, in
order to understand the Galilean Weyl tensors, it is sufficient to
consider only the electric and magnetic parts, since the full
spacetime Galilean Weyl tensor~\eqref{eq:22C_LO_uni} is uniquely given in terms of these
two parts.  (Of course, in the spatial definition, we have only
defined these parts separately.)

\subsection{Interpretation of the electric part}

As a consequence of our approach, the usual interpretation of the
electric part of the Weyl tensor as corresponding to the tidal tensor
fails for our Galilean version.  Indeed, our definition corresponds to
the \acronym{LO} term of the Lorentzian electric part, which is an order-$(-1)$
off-shell object, while the tidal tensor only appears at zeroth order,
which agrees with the \acronym{LO} term only on-shell, i.e., assuming the
Newton--Cartan equation.  Being the $(-1)$-order term, our definition
depends directly on the traceless part of the spatial Ricci tensor.
Therefore, to obtain the usual relation between the electric part and
the tidal tensor, a second electric Galilean Weyl tensor would need to
be introduced, corresponding to the next-to-leading-order.  We do not
provide such a definition in the present paper.

\enlargethispage{1.6\baselineskip}

\subsection{Interpretation of the magnetic part}
\label{sec:inter_mag}

Contrary to the electric part, the off-shell magnetic part of the Weyl
tensor depends on the choice of observer, as it depends explicitly on
the shear and vorticity tensors.  Additionally, as for the on-shell
approach (with the Newton--Cartan equation), it is the zeroth order
term of the Lorentzian magnetic part.  The classical result is that
the magnetic part vanishes in Newtonian gravity.  This result holds
on-shell with the Newton--Cartan equation.  As we show below,
off-shell this is not the case in general.

As shown in \Cref{sec:Gal_kin_lim}, the Galilean limit imposes that there exists a vector
field $w_i$ such that any vorticity tensor has the form
\begin{align}
  \label{eq:vort_decomp_for_magnetic_part}
  \Omega_{ij} = \partial_{[i} w_{j]} + \ator_{[i} w_{j]} \; ,
\end{align}
and the difference between the vorticities of two different observers
(see~\eqref{eq:diff_kin}) implies that there exist irrotational observers.

Contrary to the vorticity tensor, no constraint on the expansion
tensor arises in the Galilean limit.  In particular, it is not
guaranteed that there exists an observer with vanishing expansion
tensor, in particular for cosmological situations.  As a consequence,
the magnetic part does not vanish in general.  To understand in which
situation it does vanish, it is useful to consider the
scalar--vector--tensor (\acronym{SVT}) decomposition of symmetric
tensors.  It states that, in 3 dimensions, given a Riemannian metric
$h_{ij}$ and a symmetric tensor field $A_{ij}$, there exist a unique
scalar field $\chi$, vector field gradient $D_{(i} v_{j)}$, and
transverse-traceless (\acronym{TT}) tensor field $X_{ij}$ such that $A_{ij} = \chi
h_{ij} + D_{(i} v_{j)} + X_{ij}$.  The existence and uniqueness of
this decomposition require either a closed manifold, or specific
asymptotic conditions at infinity (e.g., asymptotic flatness) \cite{York:1973,
  Straumann:2008}.  The former corresponds to a cosmological setup, the
latter to an isolated system.  We consider one of these two situations
in the following, such that the \acronym{SVT} decomposition applies.

We apply the \acronym{SVT} decomposition to the expansion tensor
$\Theta_{ij}$, writing it as
\begin{equation}
  \label{eq:exp_decomp_for_magnetic_part}
  \Theta_{ij} = \chi h_{ij} + D_{(i} v_{j)} + X_{ij} \; .
\end{equation}
Inserting \eqref{eq:vort_decomp_for_magnetic_part} and \eqref{eq:exp_decomp_for_magnetic_part} into the definition \eqref{eq:def_H_Gal} of the magnetic part of
the Galilean Weyl tensor, one may show that it takes the form
\begin{align}
  \label{eq:H_with_decomp}
  \uH{}_{ij}
  &= \frac{1}{2}\eta^{ab}{}_{(i} \CR_{j)a} \left(v_b + w_b\right)
    + \eta^{ab}{}_{\langle i} \biggl[
      \frac{1}{2} D_a D_{j\rangle} \left(v_b - w_b\right)
      + D_a X_{j\rangle b}  \nonumber \\
  &\qquad\qquad\qquad + \frac{1}{2} D_a \left(\ator_{b} w_{j\rangle}\right)
      - \frac{1}{2} D_a\left(w_{b} \ator_{j\rangle}\right)
      + \ator_{j\rangle} \left(D_{a} w_b + \ator_{a} w_b\right)
    \biggr] .
\end{align}
Note that in this equation, the traceless part is only taken on the
indices $i$ and $j$, and not on the enclosed indices $a$, $b$ `in
between'.  In order to derive \eqref{eq:H_with_decomp}, one needs to express the commutator
of spatial covariant derivatives in terms of the spatial Riemann
tensor, and then express the Riemann tensor (which is possible in 3
dimensions) in terms of the metric and the Ricci tensor.

Again, in general, without any assumptions on the kinematical
quantities, the magnetic part \eqref{eq:H_with_decomp} does not vanish.  However, when
$\ator_i = 0$ and the Newton--Cartan equation is assumed, then
$\CR_{ij} = 0$, and for any observer we have \cite{Vigneron:2021}\footnote{This is a
  consequence of the mixed spatial-temporal projection of the
  Newton--Cartan equation.}
\begin{equation}
  \label{eq:flat_observer}
  D_i \chi = 0 \; , \quad D_i v_j = D_i w_j \; , \quad \Delta X_{ij} = 0 \; .
\end{equation}
The last condition can be derived from $\frac{1}{2} \partial_t h_{ij} =
\Theta_{ij}$ and $\partial_t \CR_{ij} = 0$, where $\partial_t$ is the time derivative
along the observer (i.e., timelike vector field) with expansion tensor
$\Theta_{ij}$.
\pagebreak

Using the boundary conditions needed for the \acronym{SVT}
decomposition theorem to hold, the last condition in~\eqref{eq:flat_observer} implies $0 =
\int_\Sigma \dd x^3 \sqrt{h} \, X^{ij} \Delta X_{ij} = -\int_\Sigma \dd x^3 \sqrt{h} \, D^k
X^{ij} D_k X_{ij}$. Then, using $D^k X^{ij} D_k X_{ij} \geq 0$, we obtain
$D_k X_{ij} = 0$.  Consequently, $X_{ij}$ is curl-free, i.e., $D_{[a}
X_{b]j} = 0$.  Therefore \eqref{eq:flat_observer} implies that the magnetic part vanishes
for any observer.

This result, obtained in the context of the Newton--Cartan equation,
suggests the following definition of a specific type of observers that
can also be considered off-shell, which we call \emph{Galilean observers}:

\begin{definition}[Galilean observers]
  \label{def:Gal_obs}
  An observer is called \emph{Galilean} if its expansion and vorticity
  tensors satisfy
  \begin{equation}
    \Theta_{ij} = H(t) h_{ij} + X_{ij} \; , \quad \Omega_{ij} = 0 \; ,
  \end{equation}
  where $X_{ij}$ is a curl-free \acronym{TT} tensor.  In other words, an
  observer is Galilean if it is vorticity-free and the \acronym{SVT}
  decomposition of its expansion tensor has no gradient part and a
  curl-free \acronym{TT} part.
\end{definition}

\begin{remark}
  In the standard approach to the Galilean limit, i.e., with $\ator_i
  = 0$, it is sometimes assumed that the decompositions \eqref{eq:vort_decomp_for_magnetic_part} and \eqref{eq:exp_decomp_for_magnetic_part} of
  the vorticity and expansion tensors feature the same vector field,
  i.e., $v_i = w_i$, and that we have $\Theta_{ij} = D_{(i} v_{j)}$ and
  $\Omega_{ij} = D_{[i} v_{j]}$ (see, e.g., \cite{Ellis:1971, Ellis:2009}).  As discussed above, this
  assumption is not generally true, for two reasons: (i)~the Galilean
  limit alone does not constrain the expansion tensor and its relation
  to the vorticity tensor, and (ii)~even assuming the Newton--Cartan
  equation, the two additional terms $H(t) h_{ij}$ and $X_{ij}$ are,
  in general, present in the expansion tensor, representing the total
  volume expansion of the spatial manifold (if finite) and the
  anisotropic expansion, respectively.  If the spatial manifold is
  $\mR^3$, i.e., assuming asymptotic flatness for the \acronym{SVT}
  decomposition theorem to hold, then we necessarily have $H(t) = 0$
  and $X_{ij} = 0$.
\end{remark}

\enlargethispage{.6\baselineskip}

The existence of Galilean observers is a key implication of the
Newton--Cartan equation.  In particular, when $X_{ij} = 0$, it ensures
that there exists a coordinate system adapted to the spatial
foliation, and such that the spatial metric is $h_{ij} = a(t)^2 \bar
h_{ij}(x^k)$, where $\dot a/a := \chi(t)$ and $\bar h_{ij}(x^k)$ is
time-independent.  This is not guaranteed without the existence of
Galilean observers.  Additionally, when the Newton--Cartan equation is
considered, the \emph{opposite} to the 4-acceleration of these observers
(with respect to a compatible torsion-free connection) is a solution
of the (cosmological) Poisson equation, and therefore corresponds to
the Newtonian gravitational field~\cite{Vigneron:2021}.  For these reasons, these
observers play a fundamental role in Newtonian gravitation.

By the definition of Galilean observers (and \eqref{eq:H_with_decomp}), we
have:
\begin{proposition}
  The magnetic Galilean Weyl tensor with respect to any Galilean
  observer vanishes.
\end{proposition}

Therefore, if one wants the magnetic part to vanish, then, rather than
considering a specific field equation, a natural (pre-field-equation,
i.e., geometric/kinematical) assumption is to rather assume the
existence of Galilean observers.  If we do so, then, for any observer,
their expansion and vorticity tensors satisfy
\begin{equation}
  \Theta_{ij}
    = H(t) h_{ij} + D_{(i} v_{j)} + X_{ij} \; , \quad
  \Omega_{ij}
    = \partial_{[i} v_{j]} + \ator_{[i} v_{j]} \; ,
\end{equation}
where $X_{ij}$ is a curl-free \acronym{TT} tensor, and the magnetic part of the
Weyl tensor is
\begin{align}
  \uH{}_{ij}
    = \eta_{ab}{}_{(i} \CR_{j)a} v_b
      + \eta^{ab}{}_{\langle i}
      \left[
        \frac{1}{2} D_a\left(\ator_{b} v_{j\rangle}\right)
        - \frac{1}{2} D_a\left(v_{b} \ator_{j\rangle}\right)
        + \ator_{j\rangle} \left(D_{a} v_b + \ator_{a} v_b\right)
      \right] .
\end{align}
Taking $v^i = 0$ corresponds to considering a Galilean observer, and
we get $H_{ij} = 0$.

\subsection{Condition for a `Galilean' theory to be a `Newtonian' theory}
\label{sec:Newtonianity}

When considering the Newton--Cartan equation alone, without deriving
the Galilean structure from a Galilean limit, the existence of
irrotational (i.e., vorticity-free) observers is not guaranteed.  For
this reason, Trautman introduced a condition on connections in
standard Newton--Cartan gravity guaranteeing the existence of
irrotational observers \cite{Trautman:1963, Trautman:1965}; later, connections satisfying
this condition were dubbed `Newtonian' connections \cite{Kuenzle:1972}.  The
condition may, for example, be written as the curvature condition
$R^\mu{}_\nu{}^\rho{}_\sigma = R^\rho{}_\sigma{}^\mu{}_\nu{}$ (see, e.g.,
\cite[Section~2.4]{Schwartz:2026}).\footnote{Actually, even when this condition is considered, only
  the existence observers with a harmonic vorticity, i.e., $\Delta \Omega_{ij} =
  0$, are guaranteed (see \cite{Dautcourt:1990a} and references therein).  Therefore,
  deriving the Galilean structure from a Galilean limit (with $\tau\wedge\dd\tau
  = 0$) is really what guarantees the existence of irrotational
  observers.}  This lead to the distinction between `Galilean geometries/theories'
and `Newtonian geometries/theories': both are constructed from
Galilean structures, but only the latter fulfils the Newtonian
condition.  More generally, a Galilean theory is considered physically
relevant if it is Newtonian.

Inspired by our discussion in the previous section, and in the spirit
of this classical `Newtonian' condition introducd by Trautman, it
seems natural to require that a Galilean theory should be called
`Newtonian' only if there exist (Galilean) observers with vanishing
magnetic part of the Galilean Weyl tensor.  The status of this
additional condition for the `Newtonianity' of a theory is, however,
not the same as Trautman's `Newtonian' condition: the latter follows
automatically in the off-shell Galilean limit from the interchange
symmetry of the Lorentzian Riemann tensor, while the former condition
is not imposed by the Galilean limit, and in this sense is added
by~hand.

\subsection{Galilean structures with spatial curvature}
\label{sec:Spatial_Curvature}

Since the off-shell electric part of the Galilean Weyl tensor features
the spatial curvature, it is worth mentioning situations where such a
curvature appears in Galilean theories.  There are two main
situations:
\begin{enumerate}
\item One may consider Galilean theories with spatial curvature to
  reproduce, at \acronym{LO}, effects that are usually obtained at
  post-Newtonian orders in the classical (spatially flat) approach,
  like the perihelion precession of Mercury (see, e.g.,
  refs.~\cite{Abramowicz.EtAl:2014, Hansen.Hartong.Obers:2019,
    Hansen.Hartong.Obers:2020}).  In this approach, the spatial curvature is
  sourced by the time dilation vector field.  This approach does not
  aim at describing new physics, but rather at reproducing known one
  from the classical limit.  In this sense, the presence of spatial
  curvature is a choice.
\item Spatial curvature also arises when the spatial manifold has a
  topology not allowing for a Euclidean, i.e., flat, metric
  \cite{Kuenzle:1976, Roukema.Rozanski:2009, Vigneron:2022,
    Vigneron.Roukema:2023, Vigneron:2024}, e.g., with $\Sigma \cong \mS^3$ or $\Sigma \cong \mS \times
  \mS^2$.  In this situation, the presence of spatial curvature is not
  a choice (contrary to the previous example), but a necessity imposed
  by topology.  It was first shown by \textcite[\emph{Remark} on p.~448]{Kuenzle:1976}\footnote{See also the early work of \textcite{Zoellner:1872}, which is discussed in ref.~\cite{Kragh:2012}.}, and
  later more systematically studied in refs.~\cite{Vigneron:2022, Vigneron.Roukema:2023, Vigneron:2024}, that,
  to get a physically consistent Newtonian theory with spatial
  curvature (of topological origin), the Newton--Cartan equation
  should feature a background Ricci curvature,
  \begin{align}
    R_{\mu\nu} - \bar R_{\mu\nu} = 4\pi G \rho \tau_\mu \tau_\nu \; ,
  \end{align}
  where $\bar R_{\mu\nu}$ is constructed from the canonical (or
  topological) Ricci curvature of the considered topology.  That
  canonical curvature is given by the Thurston classification of
  3-dimensional topological spaces.
\end{enumerate}

As discussed in \Cref{sec:Newtonianity}, in both of these cases, for the Galilean theory
to be physically relevant, or in other words `Newtonian', the
existence of Galilean observers should be demanded.  The presence of
these observers implies some conditions on the time evolution of the
spatial curvature.  Indeed, taking $\iR{G}{\partial}_t$ to be the time
derivative along a Galilean observer, i.e., $\tfrac{1}{2} \iR{G}{\partial}_t
h_{ij} = H(t) h_{ij} + X_{ij}$, we obtain
\begin{align}
  \label{eq:dt_R_ij}
  \iR{G}{\partial}_t \CR_{ij}
  &= \Delta_\textrm{L} \left(\tfrac{1}{2}\iR{G}{\partial}_t h_{ij}\right)
    - D_i D_j \left(\tfrac{1}{2}h^{kl} \iR{G}{\partial}_t h_{kl}\right)
    \nonumber\\
  &= \Delta_\textrm{L} X_{ij}
  = \Delta X_{ij}
  = 3 \CR_{k(i} X_{j)}{}^k
    - h_{ij} \CR_{kl} X^{kl}
    - \tfrac{1}{2} \CR X_{ij} \; ,
\end{align}
where $\Delta_\textrm{L}$ is the Lichnerowicz Laplacian.\footnote{The
  Lichnerowicz Laplacian is defined as $\Delta_\textrm{L} Y_{ij} := -\Delta Y_{ij} + 2
  D^k D_{(i} Y_{j)k}$ for divergence-free, symmetric $Y_{ij}$.  If in
  addition $Y_{ij}$ is curl-free, one has $\Delta_\textrm{L} Y_{ij} = \Delta Y_{ij}$
  \emph{without} a minus sign, as used in \eqref{eq:dt_R_ij} applied to $X_{ij}$.  The last
  equality in \eqref{eq:dt_R_ij} follows by a standard computation using that
  $X_{ij}$ is divergence- and curl-free, and expressing the
  3-dimensional Riemann tensor in terms of the metric and the Ricci
  tensor.}  Conversely, \eqref{eq:dt_R_ij} constrains the anisotropic expansion
tensor $X_{ij}$ if the spatial curvature is prescribed.

\enlargethispage{1.2\baselineskip}

\section{Conclusion}
\label{sec:concl}

In this paper, we have provided manifestly conformally invariant
off-shell definitions for the Galilean Weyl tensor, its electric part,
and its magnetic part.  We followed two approaches: a Galilean Weyl
tensor can be defined (i)~by using a spacetime connection (\Cref{def:Gal_Weyl_off_shell_restricted}), or
(ii)~by using the spatial connection and kinematical quantities (\Cref{def:Gal_Weyl_off_shell_spatial}).
In both approaches, the clock form was foliation forming but not
exact.

While the former approach is closer to the usual Newton--Cartan
formalism in which equations are formulated in terms of the Galilean
structure, a spacetime connection, and its curvature tensor, the
freedom in the choice of that connection (due to $\dd\tau \ne0$) makes
the definitions of the Galilean Weyl tensors not unique, and difficult
to interpret.  In \Cref{def:Gal_Weyl_off_shell_restricted}, we chose the unique boost-invariant
connection~\eqref{eq:unique_A_con} that can be derived from a Galilean structure and a
choice of timelike vector field and Coriolis field.  But, as discussed
in \Cref{sec:discuss_spacetime_Weyl}, other connections could be considered, and therefore, this
remains a choice.

The purely spatial definition (\Cref{def:Gal_Weyl_off_shell_spatial}) of the electric and magnetic parts
of the Galilean Weyl tensor does not suffer from the issues of the
spacetime approach because there is a unique spatial connection to
consider, namely the Levi-Civita connection of the spatial metric.  In
this approach, the electric and magnetic parts are expressed in terms
of the kinematical quantities of an observer, making these definitions
much more geometrically transparent.  Contrary to the on-shell
approach, the electric part is not related to the tidal tensor but to
the spatial curvature, and the magnetic part does not vanish in
general.

In the spirit of the `Newtonian' condition for standard Newton--Cartan
gravity, introduced by Trautman, we argued that requiring the
existence of specific observers for which the magnetic part vanishes,
called \emph{Galilean observers}, should be a necessary condition for a
Galilean theory to be considered physical, hence called `Newtonian'.
We argued that this condition should be considered when constructing
any Galilean invariant theory featuring spatial curvature.

\medskip

Stepping back, our work helps to clarify and contextualise existing
results for the Galilean Weyl tensors presented in refs.~\cite{Ehlers.Buchert:2009, Dewar.Weatherall:2018,
  Dewar.Read:2020}, and to understand their limitations when they are
considered off-shell.  But---even more importantly---it should pave
the way for future directions of research, for example:
\begin{enumerate}
\item With a fully general understanding of Galilean Weyl tensors in hand,
  it is natural to ask whether this object can be leveraged in ways
  analogous to the relativistic context---for example, in developing a
  non-relativistic analogue of the Petrov classification of
  spacetimes, or in understanding non-relativistic gravitational
  radiation (on which see ref.~\cite{Linnemann.Read:2021}).
\item One can ask whether properties of the relativistic Weyl tensor carry
  over to the Galilean regime.  For example, does a vanishing Galilean
  Weyl tensor imply conformal flatness of the Galilean structure?
\item Since there has been much recent work on Galilean conformal field
  theory \cite{Bagchi.Gopakumar:2009, Bagchi.Mandal:2009} and twistor theory \cite{Dunajski.Gundry:2016, Dunajski.Penrose:2023, March.Read:2026}, it is
  natural to ask whether the Galilean Weyl tensor will constitute an
  object of study in those contexts going forward.
\end{enumerate}

\section*{Acknowledgements}

P.K.S.\ and Q.V.\ thank Pembroke College, Oxford, for hospitality
during a stay where the present research was conducted.  Q.V.\ was
supported in part by the European Research Council (\acronym{ERC})
under the European Union's Horizon 2020 research and innovation
program (grant agreement \acronym{ERC} advanced grant
740021-\acronym{ARTHUS}, principal investigator Thomas Buchert).  For
helpful discussions J.R.\ thanks Niels Linnemann, Eleanor March, and
the audience of a talk at \acronym{UC} Irvine.

\appendix

\section{Proof of uniqueness of the boost-invariant connection}
\label{app:proof_unique_A_con}

Here we prove the uniqueness result for the boost-invariant connection
discussed in \Cref{sec:unique_A_con}.

First, we prove a very helpful statement:

\enlargethispage{1.4\baselineskip}

\begin{proposition}
  \label{prop:boost-inv}
  Let $(\tau_\mu, h^{\mu\nu})$ be a Galilean structure.  A connection is
  boost-invariant if and only if its non-metricities with respect to
  $(\tau_\mu, h^{\mu\nu})$ and torsion are boost-invariant and its Coriolis
  field transforms under boosts according to \eqref{eq:fix_con_Coriolis_trafo}, i.e., satisfies
  \begin{align}
    \label{eq:boost-inv_Coriolis_trafo}
    \ukappa{}_{\mu\nu}
    &\xrightarrow{u \to u - v} \ukappa{}_{\mu\nu} - \partial_{[\mu} v_{\nu]}
      + \tau_{[\mu} \partial_{\nu]} \left(\tfrac{v^2}{2}\right)
      + \frac{1}{2} T^\rho{}_{\mu\nu} v_\rho \nonumber\\
    &\qquad\qquad- Q_{[\mu\nu]\rho} v^\rho + \tfrac{1}{2} \tau_{[\mu} Q_{\nu]\rho\sigma} v^\rho v^\sigma
      + (u^\rho - v^\rho) \bigl(v_{[\mu} + v^2 \tau_{[\mu}\bigr) Q_{\nu]\rho} \; .
  \end{align}

  \begin{proof}
    The proof of this statement is essentially an exercise in
    precisely keeping track of which objects depend on a choice of
    (unit) timelike vector field (with respect to $(\tau_\mu, h^{\mu\nu})$).
    We will carry this out in a very explicit way, such as to be as
    clear as possible.

    For the sake of explicitness, for the period of this proof we
    introduce the notation
    \begin{equation}
      \kappa[u;\nabla]_{\mu\nu}
      := (\text{Coriolis field of $\nabla$ with respect to $u$})_{\mu\nu}
      = (\nabla_{[\mu} u^\alpha) \ub{u}_{\nu]\alpha}
    \end{equation}
    for any affine connection $\nabla$ and timelike vector field $u^\mu$.
    According to the discussion from \Cref{sec:Gal_con}, the Coriolis fields of a
    fixed connection with respect to different timelike vector fields
    satisfy
    \begin{align}
      \label{eq:boost_inv_proof_gen_Coriolis_trafo}
      \kappa[u - v;\nabla]_{\mu\nu}
      &= \kappa[u;\nabla]_{\mu\nu}
        - \partial_{[\mu} v_{\nu]}
        + \tau_{[\mu} \partial_{\nu]} \left(\tfrac{v^2}{2}\right)
        + \tfrac{1}{2} T^\rho{}_{\mu\nu} v_\rho \nonumber\\
      &\quad- Q_{[\mu\nu]\rho} v^\rho + \tfrac{1}{2} \tau_{[\mu} Q_{\nu]\rho\sigma} v^\rho v^\sigma
        + (u^\rho - v^\rho) \bigl(v_{[\mu} + v^2 \tau_{[\mu}\bigr) Q_{\nu]\rho}
    \end{align}
    for $u^\mu$ a timelike and $v^\mu$ a spacelike vector field, where
    $Q_\alpha{}^{\mu\nu} = \nabla_\alpha h^{\mu\nu}$, $Q_{\mu\nu} = \nabla_\mu \tau_\nu$ are the
    non-metricities and $T^\alpha{}_{\mu\nu}$ is the torsion of $\nabla$.

    Now we are going to more precisely formulate the statement we want
    to prove.  Explicitly stated in fully written-out form, it means
    the following: for a map $u \mapsto \nabla[u]$ from the set of timelike
    vector fields to the set of affine connections, we have
    \begin{equation}
      \label{eq:boost-inv_proof}
      \nabla[u] = \nabla[\tilde{u}] \; \text{for all timelike vector fields} \; u^\mu, \tilde{u}^\mu
    \end{equation}
    (i.e., the map is constant) if and only if we have
    \begin{subequations} \label{eq:boost_inv_proof_conds}
    \begin{equation}
      \label{eq:boost_inv_proof_cond_non-metr_tor}
      \nabla[u]_\mu \tau_\nu = \nabla[\tilde{u}]_\mu \tau_\nu \; , \quad
      \nabla[u]_\alpha h^{\mu\nu} = \nabla[\tilde{u}]_\alpha h^{\mu\nu} \; , \quad
      T[\nabla[u]]^\alpha{}_{\mu\nu} = T[\nabla[\tilde{u}]]^\alpha{}_{\mu\nu}
    \end{equation}
    for all timelike vector fields $u^\mu, \tilde{u}^\mu$ and
    \begin{align}
      \label{eq:boost_inv_proof_cond_Coriolis}
      \kappa[u - v;\nabla[u - v]]_{\mu\nu}
      &= \kappa[u;\nabla[u]]_{\mu\nu}
        - \partial_{[\mu} v_{\nu]}
        + \tau_{[\mu} \partial_{\nu]} \left(\tfrac{v^2}{2}\right)
        + \tfrac{1}{2} T^\rho{}_{\mu\nu} v_\rho \nonumber\\
      &\quad- Q_{[\mu\nu]\rho} v^\rho + \tfrac{1}{2} \tau_{[\mu} Q_{\nu]\rho\sigma} v^\rho v^\sigma
        + (u^\rho - v^\rho) \bigl(v_{[\mu} + v^2 \tau_{[\mu}\bigr) Q_{\nu]\rho}
    \end{align}
    \end{subequations}
    for any timelike vector field $u^\mu$ and spacelike vector field
    $v^\mu$, where the non-metricities and torsion are independent of
    $u^\mu$ by \eqref{eq:boost_inv_proof_cond_non-metr_tor}.

    Now we are going to prove the claimed equivalence of \eqref{eq:boost-inv_proof} and \eqref{eq:boost_inv_proof_conds}.
    First assuming that the connection is boost-invariant, i.e., $\nabla[u]
    = \nabla$ is a fixed connection, its non-metricities and torsion are
    independent of $u$, i.e., we have \eqref{eq:boost_inv_proof_cond_non-metr_tor}; further, the equation \eqref{eq:boost_inv_proof_gen_Coriolis_trafo}
    relating the Coriolis fields of a fixed connection with respect to
    different timelike vector fields directly translates into \eqref{eq:boost_inv_proof_cond_Coriolis}.

    For the converse, assume that the equations \eqref{eq:boost_inv_proof_conds} hold.  We fix a
    timelike vector field $u^\mu$ and a spacelike vector field $v^\mu$,
    and consider the two connections $\nabla[u]$ and $\nabla[u - v]$.  By \eqref{eq:boost_inv_proof_cond_non-metr_tor},
    we know that these have agreeing non-metricities and torsion.
    Further, combining the general equation \eqref{eq:boost_inv_proof_gen_Coriolis_trafo} applied to $\nabla =
    \nabla[u-v]$ with the assumed equation \eqref{eq:boost_inv_proof_cond_Coriolis} yields
    \begin{equation}
      \kappa[u;\nabla[u-v]]_{\mu\nu} = \kappa[u;\nabla[u]]_{\mu\nu} \; .
    \end{equation}
    Therefore, the two connections $\nabla[u]$ and $\nabla[u - v]$ have agreeing
    non-metricities, torsion, and Coriolis field with respect to $u$.
    Therefore, according to the classification theorem~\cite{Schwartz:2025} they agree,
    $\nabla[u] = \nabla[u - v]$.  Thus we have shown that $\nabla[u]$ is independent
    of $u$, finishing the proof.
  \end{proof}
\end{proposition}

Now we can prove uniqueness and existence of the boost-invariant
connection discussed in \Cref{sec:unique_A_con}.

\begin{proof}[Proof of \Cref{prop:unique_A_con}]
  Let $(\tau_\mu, h^{\mu\nu})$ be a Galilean structure with $\tau\wedge\dd\tau = 0$, $u^\mu$
  a timelike vector field, and $\ukappa{}_{\mu\nu}$ an arbitrary two-form.

  We consider an affine connection $\nabla$ for which we fix $(\nabla_{[\mu} u^\sigma)
  \ub{}_{\nu]\sigma} = \ukappa{}_{\mu\nu}$, and whose torsion and non-metricities
  are defined locally from $(\tau_\mu, h^{\mu\nu})$ and $u^\mu$, depending on
  derivatives of $\tau_\mu$, $h^{\mu\nu}$ only to first order and linearly.
  Consequently, the torsion and non-metricities are tensor fields
  depending pointwise on $(\tau_\mu, h^{\mu\nu})$ and $u^\mu$, and pointwise
  linearly on $(\dd\tau)_{\mu\nu}$---i.e., on $\ator^\mu$ as we assume $\tau\wedge\dd\tau
  = 0$---and on $\Lie{u} h^{\mu\nu}$.  However, regarding $\Lie{u}
  h^{\mu\nu}$, only its spatial part $-\frac{1}{2} \ub{}_{\mu\alpha} \ub{}_{\nu\beta}
  \Lie{u} h^{\alpha\beta} = \frac{1}{2} \Lie{u} \ub{}_{\mu\nu} = \uTheta{u}_{\mu\nu}$
  is independent of $\ator^\mu$.\footnote{We have the identity $\Lie{u} h^{\mu\nu} =
    -h^{\mu\alpha} h^{\nu\beta} \Lie{u} \ub{}_{\alpha\beta} - 2 \ator^{(\mu} v^{\nu)}$, showing
    that the purely temporal part of $\Lie{u} h^{\mu\nu}$ vanishes, and
    that its mixed temporal-spatial part depends on $\ator^\mu$.}

  Grouping together the torsion and non-metricity terms, we therefore
  consider a connection of the form
  \begin{equation}
    \Gamma^\alpha_{\mu\nu}
    = \ucheckGamma{u}^\alpha_{\mu\nu} + 2 \tau_{(\mu} \ukappa{}_{\nu)}{}^\alpha
      + \CC^\alpha{}_{\mu\nu}[\ator^\mu, \Theta_{\mu\nu}] \; ,
  \end{equation}
  where $\ucheckGamma{}^\alpha_{\mu\nu}$ is defined in~\eqref{eq:con_u} and
  $\CC^\alpha{}_{\mu\nu}[\ator^\mu, \Theta_{\mu\nu}]$ is a (1,2)-tensor depending
  pointwise on $(\tau_\mu, h^{\mu\nu})$, on $u^\mu$, and linearly on $\ator^\mu$
  and $\Lie{u} \ub{}_{\mu\nu}$.  Therefore, the only possible independent
  terms in $\CC^\alpha{}_{\mu\nu}$ are
  \begin{align}
    \label{eq:proof_unique_A_con_indep_terms}
    \ub{}^\alpha{}_{(\mu} \uator{}_{\nu)} \; ,
    && \ub{}_{\mu\nu} \ator^\alpha \; ,
    && u^\alpha \tau_{(\mu} \uator{}_{\nu)} \; ,
    && \tau_\mu \tau_\nu \ator^\alpha \; ,
    && \ub{}^\alpha{}_{[\mu} \uator{}_{\nu]} \; ,
    && u^\alpha \tau_{[\mu} \uator{}_{\nu]} \; , \nonumber\\
    \ub{}^\alpha{}_{(\mu} \tau_{\nu)} \, \theta \; ,
    && \ub{}_{\mu\nu} u^\alpha \, \theta \; ,
    && u^\alpha \tau_\mu \tau_\nu \, \theta \; ,
    && \ub{}^\alpha{}_{[\mu} \tau_{\nu]} \, \theta \; , \nonumber\\
    \Theta^\alpha{}_{(\mu} \tau_{\nu)} \; ,
    && \Theta_{\mu\nu} u^\alpha \; ,
    && \Theta^\alpha{}_{[\mu} \tau_{\nu]} \; ,
  \end{align}
  where $\theta := \Theta_{\mu\nu} h^{\mu\nu}$.\footnote{A similar approach was followed in ref.~\cite{Vigneron.Barzegar.Read:2025} to define `reduced non-metricities' and `reduced torsion' in order to remove the appearance of $\omega_{\mu\nu}$ in the identities~\eqref{eq:Gal_Q_T}.  However, terms with $\Theta_{\mu\nu}$ were not considered.  Furthermore, a term of the form `$\ator_\alpha \, h^{\mu\nu}$' was omitted when listing the possible terms in the non-metricity $Q_\alpha{}^{\mu\nu}$ (see Appendix~A.1 in ref.~\cite{Vigneron.Barzegar.Read:2025}).  As a result, the proposed construction in that paper is not exhaustive and therefore not compatible with the study of the present
    paper.}%
  \textsuperscript{,}\footnote{Replacing $\Theta_{\mu\nu}$ by $\sigma_{\mu\nu} := \Theta_{\mu\nu} -
    \frac{\theta}{3} \ub{}_{\mu\nu}$ is also possible, but equivalent.
    However, using the former is simpler for the proof.}

  As we fix the Coriolis field with respect to $u^\mu$ with $(\nabla_{[\mu}
  u^\alpha) \gamma_{\nu]\alpha} = \ukappa{}_{\mu\nu}$, this implies the constraint
  \begin{equation}
    u^\beta \ub{}_{\alpha[\mu} \, \CC^\alpha{}_{\nu]\beta} = 0 \; .
  \end{equation}
  Therefore, the term $\tau_\mu \tau_\nu \ator^\alpha$ is forbidden.  Hence, the
  connection has the general form\footnote{The terms $\ub{}_{\mu\nu} u^\alpha$ and
    $\Theta_{\mu\nu} u^\alpha$ are not dimensionless, contrary to the other terms
    in~\eqref{eq:con_Gal_ator}.  Therefore, the coefficients $d_2$ and $d_6$ must have the
    dimension of a velocity to the power $-2$.}
  \begin{align}
    \label{eq:con_Gal_ator}
    \Gamma^\alpha_{\mu\nu}
    &= \ucheckGamma{}^\alpha_{\mu\nu} + 2\tau_{(\mu} \ukappa{}_{\nu)}{}^\alpha \nonumber\\
    &\quad+ c_1 \ub{}^\alpha{}_{(\mu} \uator{}_{\nu)}
      + c_2 \ub{}_{\mu\nu} \ator^\alpha
      + c_3 \, u^\alpha \tau_{(\mu} \uator{}_{\nu)}
      + c_4 \ub{}^\alpha{}_{[\mu} \uator{}_{\nu]}
      + c_5 \, u^\alpha \tau_{[\mu} \uator{}_{\nu]} \nonumber\\
    &\quad+ \left(
      d_1 \ub{}^\alpha{}_{(\mu} \tau_{\nu)}
      + d_2 \ub{}_{\mu\nu} u^\alpha
      + d_3 \, u^\alpha \tau_\mu \tau_\nu
      + d_4 \ub{}^\alpha{}_{[\mu} \tau_{\nu]}
      \right) \theta \nonumber\\
    &\quad+ d_5 \, \Theta^\alpha{}_{(\mu} \tau_{\nu)}
      + d_6 \, \Theta_{\mu\nu} u^\alpha
      + d_7 \, \Theta^\alpha{}_{[\mu} \tau_{\nu]} \; ,
  \end{align}
  where the $c_i$ and $d_i$ are constants.  The non-metricities and
  torsion of this connection can be computed to be
  \begin{subequations} \label{eq:con_Gal_ator_non-met_time_full}
  \begin{align}
    \label{eq:con_Gal_ator_non-met_time}
    Q_{\mu\nu}
    &= -(1 - c_5) \uator{}_{[\mu} \tau_{\nu]}
      - c_3 \uator{}_{(\mu} \tau_{\nu)}
      - (d_2 \ub{}_{\mu\nu} + d_3 \, \tau_\mu \tau_\nu) \, \theta
      - d_6 \, \Theta_{\mu\nu} \; , \\
    \label{eq:con_Gal_ator_non-met_space}
    Q_\alpha{}^{\mu\nu}
    &= (c_1 - c_4) \uator{}_\alpha \, h^{\mu\nu}
      + (-1 + c_3 + c_5) \, \tau_\alpha u^{(\mu} \ator^{\nu)}
      + (c_1 + 2 c_2 + c_4) \ub{}^{(\mu}{}_\alpha \ator^{\nu)} \nonumber\\
    &\quad+ \left[(d_1 - d_4) \, \tau_\alpha h^{\mu\nu}
        + 2 d_2 \ub{}^{(\mu}{}_\alpha u^{\nu)} \right] \theta
      + (d_5 - d_7) \, \tau_\alpha \Theta^{\mu\nu}
      + 2 d_6 \, \Theta^{(\mu}{}_\alpha u^{\nu)} \; , \\
    \label{eq:con_Gal_ator_tor}
    T^\alpha{}_{\mu\nu}
    &= 2 c_4 \ub{}^\alpha{}_{[\mu} \uator{}_{\nu]}
      + 2 c_5 \, u^\alpha \tau_{[\mu} \uator{}_{\nu]}
      + 2 d_4 \, \delta^\alpha_{[\mu} \tau^{\vphantom{\alpha}}_{\nu]} \theta
      + 2 d_7 \, \Theta^\alpha{}_{[\mu} \tau_{\nu]} \; ,
  \end{align}
  \end{subequations}
  where we have used $\ub{}^\alpha{}_{[\mu} \tau_{\nu]} = \delta^\alpha_{[\mu}
  \tau^{\vphantom{\alpha}}_{\nu]}$.

  Now we consider how the different quantities appearing in \eqref{eq:proof_unique_A_con_indep_terms}
  transform under Galilean boosts $u^\mu \to u^\mu - v^\mu$.  We denote the
  \emph{linearised} change of any quantity $f[u]$ under this boost by $\delta_v
  f[u]$, i.e.,
  \begin{equation}
    f[u - v] =: f[u] + \delta_v f[u] + \text{(terms at least quadratic in
      $v$)}.\footnote{In fact, for the quantities considered in the following,
            only quadratic terms would appear.}
  \end{equation}
  Using this notation, we have
  \begin{subequations} \label{eq:linear_boost_variations}
  \begin{align}
    \delta_v u^\mu
    &= -v^\mu \; , \quad
    \delta_v \ub{}_{\mu\nu}
    = 2 \tau_{(\mu} v_{\nu)} \; , \quad
    \delta_v \ub{}^\mu{}_\nu
    = v^\mu \tau_\nu \; , \quad
    \delta_v \uator{}_\mu
    = \tau_\mu \, (\ator \cdot v) \; , \\
    \delta_v \Theta_{\mu\nu}
    &= v_{(\mu} \ator_{\nu)}
      - 2 \tau_{(\mu} \iR{v}{\omega}_{\nu)\alpha} u^\alpha
      - \frac{1}{2} \Lie{v} \ub{}_{\mu\nu} \; , \quad
    \delta_v \theta
    = \ator \cdot v - \frac{1}{2} h^{\mu\nu} \Lie{v} \ub{}_{\mu\nu} \; ,
  \end{align}
  \end{subequations}
  where we defined
  \begin{equation}
    v_\mu := \ub{}_{\mu\nu} v^\nu \; , \quad
    \iR{v}{\omega}_{\mu\nu} := \partial_{[\mu} v_{\nu]} \; , \quad
    \ator \cdot v := \ub{}_{\mu\nu} \ator^\mu v^\nu \; .
  \end{equation}

  For the connection to be boost-invariant, by \Cref{prop:boost-inv} it is necessary
  that its non-metricities and torsion are boost-invariant; in
  particular, their linear boost variations $\delta_v Q_{\mu\nu}$, $\delta_v
  Q_\alpha{}^{\mu\nu}$, $\delta_v T^\alpha{}_{\mu\nu}$ need to be zero.  Additionally,
  because these need to vanish for \emph{any} spacelike vector field $v^\mu$,
  if two terms appearing in one of these expressions differ as
  functionals of $v^\mu$, then they need to vanish independently.  Using
  the linearised boost, rather than the full boost, will simplify the
  computations, and, as we will see below, this will be sufficient to
  prove \Cref{prop:unique_A_con}.

  First computing the linear boost variation of $Q_{\mu\nu}$ as given by
  \eqref{eq:con_Gal_ator_non-met_time}, using \eqref{eq:linear_boost_variations} we obtain
  \begin{align}
    \delta_v Q_{\mu\nu}
    &= -(1 - c_5) \, (\delta_v \uator{}_{[\mu}) \, \tau_{\nu]}
      - c_3 \, (\delta_v \uator{}_{(\mu}) \, \tau_{\nu)}
      - d_2 \, (\delta_v \ub{}_{\mu\nu}) \, \theta \nonumber\\
    &\quad- (d_2 \ub{}_{\mu\nu} + d_3 \, \tau_\mu \tau_\nu) \, \delta_v \theta
      - d_6 \, \delta_v \Theta_{\mu\nu} \nonumber\\
    &= -(c_3 + d_3) \, (\ator \cdot v) \, \tau_\mu \tau_\nu
      - 2 d_2 \, \tau_{(\mu} v_{\nu)} \, \theta
      - d_2 \, (\ator \cdot v) \ub{}_{\mu\nu} \, \nonumber\\
    &\quad+ (d_2 \ub{}_{\mu\nu} + d_3 \, \tau_\mu \tau_\nu) \,
        \tfrac{1}{2} h^{\rho\sigma} \Lie{v} \ub{}_{\rho\sigma}
      - d_6 \, (v_{(\mu} \ator_{\nu)} - 2 \tau_{(\mu} \iR{v}{\omega}_{\nu)\alpha} u^\alpha
        - \tfrac{1}{2} \Lie{v} \ub{}_{\mu\nu}).
  \end{align}
  This must vanish for all spacelike vector fields $v^\mu$.  Considering
  individually the coefficients of terms depending functionally on
  $v^\mu$ in different ways, this imposes
  \begin{itemize}[nosep]
  \item $d_2 = 0$ as coefficient of $\tau_{(\mu} v_{\nu)} \theta$,
  \item $d_3 = 0$ as coefficient of $\tau_\mu \tau_\nu \tfrac{1}{2} h^{\rho\sigma} \Lie{v}
    \ub{}_{\rho\sigma}$,
  \item $d_6 = 0$ as coefficient of $\tfrac{1}{2} \Lie{v} \ub{}_{\mu\nu}$, and finally
  \item $c_3 + d_3 = 0$ as coefficient of $(\ator \cdot v) \, \tau_\mu \tau_\nu$, i.e.,
    (with $d_3 = 0$ from above) $c_3 = 0$.
  \end{itemize}

  Next, for the variation of $T^\alpha{}_{\mu\nu}$ as given by \eqref{eq:con_Gal_ator_tor} we obtain
  \begin{align}
    \delta_v T^\alpha{}_{\mu\nu}
    &= 2 c_4 (\delta_v \ub{}^\alpha{}_{[\mu}) \uator{}_{\nu]}
      + 2 c_4 \ub{}^\alpha{}_{[\mu} \, \delta_v \uator{}_{\nu]}
      + 2 c_5 \, (\delta_v u^\alpha) \, \tau_{[\mu} \uator{}_{\nu]}
      + 2 c_5 \, u^\alpha \tau_{[\mu} \, \delta_v \uator{}_{\nu]} \nonumber\\
    &\quad+ 2 d_4 \, \delta^\alpha_{[\mu} \tau^{\vphantom{\alpha}}_{\nu]} \delta_v \theta
      + 2 d_7 \, h^{\alpha\beta} (\delta_v \Theta_{\beta[\mu}) \tau_{\nu]} \nonumber\\
    &= 2 (c_4 - c_5 - \tfrac{1}{2} d_7) \, v^\alpha \tau_{[\mu} \uator{}_{\nu]}
      + 2 (c_4 + d_4) \, (\ator \cdot v) \, \delta^\alpha_{[\mu} \tau^{\vphantom{\alpha}}_{\nu]}
      \nonumber\\
    &\quad- d_4 \, \delta^\alpha_{[\mu} \tau^{\vphantom{\alpha}}_{\nu]}
        h^{\rho\sigma} \Lie{v} \ub{}_{\rho\sigma}
      + d_7 \, (\ator^\alpha v_{[\mu} - h^{\alpha\beta} \Lie{v} \ub{}_{\beta[\mu}) \tau_{\nu]} \;
      .
  \end{align}
  Again looking at the coefficients of terms depending differently on
  $v^\mu$, vanishing of this expression imposes
  \begin{itemize}[nosep]
  \item $d_4 = 0$ as coefficient of $\delta^\alpha_{[\mu} \tau^{\vphantom{\alpha}}_{\nu]} h^{\rho\sigma}
    \Lie{v} \ub{}_{\rho\sigma}$,
  \item $d_7 = 0$ as coefficient of $h^{\alpha\beta} \Lie{v} \ub{}_{\beta[\mu} \tau_{\nu]}$,
  \item $c_4 - c_5 - \tfrac{1}{2} d_7 = 0$ as coefficient of $v^\alpha \tau_{[\mu}
    \uator{}_{\nu]}$, i.e., (with $d_7 = 0$ from above) $c_4 = c_5$,
    and\looseness-1
  \item $c_4 + d_4 = 0$ as coefficient of $(\ator \cdot v) \, \delta^\alpha_{[\mu}
    \tau^{\vphantom{\alpha}}_{\nu]}$, i.e., (with $d_4 = 0$ from above) $c_4 =
    0$, further implying $c_5 = 0$.
  \end{itemize}

  Finally, using all the already established vanishing of some of the
  constants $c_i$ and $d_i$, for the linear boost variation of
  $Q_\alpha{}^{\mu\nu}$ as given by \eqref{eq:con_Gal_ator_non-met_space} we obtain
  \begin{align}
    \delta_v Q_\alpha{}^{\mu\nu}
    &= c_1 \, (\delta_v \uator{}_\alpha) \, h^{\mu\nu}
      - \tau_\alpha (\delta_v u^{(\mu}) \, \ator^{\nu)}
      + (c_1 + 2 c_2) (\delta_v \ub{}^{(\mu}{}_\alpha) \, \ator^{\nu)} \nonumber\\
    &\quad+ d_1 \, \tau_\alpha h^{\mu\nu} \delta_v \theta
      + d_5 \, \tau_\alpha h^{\mu\rho} h^{\nu\sigma} \delta_v \Theta_{\rho\sigma} \nonumber\\
    &= c_1 \, (\ator \cdot v) \, \tau_\alpha h^{\mu\nu}
      + (1 + c_1 + 2 c_2 + d_5) \, \tau_\alpha v^{(\mu} \ator^{\nu)} \nonumber\\
    &\quad+ d_1 \, \tau_\alpha h^{\mu\nu}
        (\ator \cdot v - \tfrac{1}{2} h^{\rho\sigma} \Lie{v} \ub{}_{\rho\sigma})
      - d_5 \, \tau_\alpha \tfrac{1}{2} h^{\mu\rho} h^{\nu\sigma} \Lie{v} \ub{}_{\rho\sigma} \; .
  \end{align}
  Vanishing of this expression imposes
  \begin{itemize}[nosep]
  \item $d_1 = 0$ as coefficient of $\tau_\alpha h^{\mu\nu} \tfrac{1}{2} h^{\rho\sigma}
    \Lie{v} \ub{}_{\rho\sigma}$,
  \item $d_5 = 0$ as coefficient of $\tau_\alpha \tfrac{1}{2} h^{\mu\rho} h^{\nu\sigma}
    \Lie{v} \ub{}_{\rho\sigma}$,
  \item $c_1 = 0$ as coefficient of $(\ator \cdot v) \, \tau_\alpha h^{\mu\nu}$, and
  \item $1 + c_1 + 2 c_2 + d_5 = 0$ as coefficient of $\tau_\alpha v^{(\mu}
    \ator^{\nu)}$, i.e., (with $d_5 = 0$, $c_1 = 0$ from above) $c_2 =
    -\tfrac{1}{2}$.
  \end{itemize}

  Thus we have shown that for the connection \eqref{eq:con_Gal_ator} to be
  boost-invariant, we need $c_2 = -\tfrac{1}{2}$ and all the other
  constants to be zero. This shows uniqueness of boost-invariant
  connections depending linearly on first derivatives of $\tau_\mu$ and
  $h^{\mu\nu}$: any such connection has coefficients
  \begin{equation}
    \unicon{\Gamma}^\alpha_{\mu\nu}
    = \ucheckGamma{}^\alpha_{\mu\nu} + 2\tau_{(\mu} \ukappa{}_{\nu)}{}^\alpha
      - \frac{1}{2} \ub{}_{\mu\nu} \ator^\alpha ,
  \end{equation}
  with non-metricities and torsion
  \begin{equation}
    \unicon{Q}_{\mu\nu} = \tau_{[\mu} \uator{}_{\nu]} = \omega_{\mu\nu} \; , \quad
    \unicon{Q}_\alpha{}^{\mu\nu} = -\delta^{(\mu}_\alpha \ator^{\nu)}_{\vphantom{\alpha}} \; , \quad
    \unicon{T}^\alpha{}_{\mu\nu} = 0 \; .
  \end{equation}
  Note that these are boost-invariant not only to linear order (as we
  used in the proof), but exactly.  Inserting the non-metricities and
  torsion into the general transformation law \eqref{eq:boost-inv_Coriolis_trafo} for the Coriolis
  field of boost-invariant connections, a direct computation shows
  that this simplifies to
  \begin{equation}
    \ukappa{}_{\mu\nu}
    \xrightarrow{u \to u - v}
    \ukappa{}_{\mu\nu}
    - \partial_{[\mu} \left(v_{\nu]} + \tau_{\nu]} \tfrac{v^2}{2} \right),
  \end{equation}
  as claimed in the statement of \Cref{prop:unique_A_con}.

  Conversely, since the non-metricities and torsion are
  boost-invariant, it follows by \Cref{prop:boost-inv} that the unique connection we
  derived \emph{is} indeed boost-invariant, which shows existence.
\end{proof}

\begin{remark}
  In \Cref{app:restricted_boosts}, we discuss connections only invariant under restricted
  boosts, for which $\ator \cdot v = 0$.  This leaves more freedom for the
  connection, as several of the constraints on the constants $c_i$ and
  $d_i$ are not imposed.
\end{remark}

\section{Boost-invariant connections depending on extra structures}
\label{app:boost_inv_extra}

In this appendix we detail three methods to define boost-invariant
connections which all require the introduction of extra structures on
top of the Galilean structure and a choice of Coriolis field.

\subsection{With restricted boosts}
\label{app:restricted_boosts}

When proving the uniqueness of the boost-invariant connection~\eqref{eq:unique_A_con} in
\Cref{app:proof_unique_A_con}, we used at several points that for a generic boost, we have
$\ator \cdot  v \ne 0$.  A possibility to obtain another connection
than~\eqref{eq:unique_A_con} is to demand the connection to be only invariant under
\emph{restricted Galilean boosts} with $\ator \cdot v = 0$, i.e., for which
the boost parameter field $v^\mu$ is orthogonal to the time dilation
vector field.

In this case, repeating the analysis of \Cref{app:proof_unique_A_con}, we obtain the following:
\begin{itemize}
\item Vanishing of $\delta_v Q_{\mu\nu}$ still yields $d_2 = d_3 = d_6 = 0$, but we
  may have $c_3 \ne 0$;
\item vanishing of $\delta_v T^\alpha{}_{\mu\nu}$ still yields $d_4 = d_7 = 0$ and $c_4
  = c_5$, but we may have $c_4 \ne 0$; and
\item vanishing of $\delta_v Q_\alpha{}^{\mu\nu}$, where we now need to keep the $c_3$
  and $c_4 = c_5$ terms, still yields $d_1 = d_5 = 0$, but we may have
  $c_1 \ne 0$, and the last equation between the coefficients becomes
  $1 + c_1 + 2 c_2 - c_3 = 0$.
\end{itemize}
Combined, we only need
\begin{equation}
  c_3 = 1 + c_1 + 2 c_2 \; , \quad c_5 = c_4 \; ,
\end{equation}
leading to
\begin{subequations}
\begin{align}
  Q_{\mu\nu}
  &= (1-c_4) \, \tau_{[\mu} \uator{}_{\nu]}
    - (1 + c_1 + 2 c_2 ) \uator{}_{(\mu} \tau_{\nu)}\; , \\
  Q_\alpha{}^{\mu\nu}
  &=  \left(c_1-c_4\right) \uator{}_\alpha h^{\mu\nu}
    + \left(c_1 + 2c_2 + c_4\right) \delta^{(\mu}_\alpha \ator_{\vphantom{\alpha}}^{\nu)}
    \; , \\
  T^\alpha{}_{\mu\nu}
  &= 2\, c_4 \delta^\alpha_{[\mu} \uator{}^{\vphantom{\alpha}}_{\nu]} \; .
\end{align}
\end{subequations}
Since $\uator{}_\mu$ is invariant under restricted boosts, the
non-metricities and torsion are invariant under restricted boosts.
The transformation of the Coriolis field may be shown (by direct
calculation) to be
\begin{equation}
  \ukappa{}_{\mu\nu}
  \xrightarrow{u \to u - v}
    \ukappa{}_{\mu\nu}
    + \left((1 + 2 c_2) \uator{}_{[\mu} - \partial_{[\mu}\right)
      \left(v_{\nu]} + \tau_{\nu]} \tfrac{v^2}{2} \right).
  \end{equation}

Given a timelike vector field $u^\mu$, the restricted Galilean boosts
define a preferred (2-dimensional) family of timelike vector fields
whose members are all related to $u^\mu$ by a restricted boost.
Therefore, while in this approach to define a (restricted)
boost-invariant connection there is no explicit extra structure,
preferred observers are implicitly introduced.

Other restrictions than $\ator \cdot v \ne 0$ could be considered.  An
example is $D^\mu v^\nu = 0$ which defines a family of observers related
by a `spatially constant' boost.  But it is not clear if this would
produce a more general connection than the connection~\eqref{eq:unique_A_con}.

\subsection{With a mass gauge field}
\label{app:con_Mass}

Another possibility to obtain a boost-invariant connection from $(\tau_\mu,
h^{\mu\nu}, u^\mu, \ukappa{}_{\mu\nu})$ and linear in the first derivatives of
$(\tau_\mu, h^{\mu\nu})$ is to introduce a \emph{(mass) gauge field} / Bargmann form
$M_\mu$ such that \cite{Jensen:2018, Hartong.Kiritsis.Obers:2015, Bekaert.Morand:2016, Festuccia.EtAl:2016, Hansen.Hartong.Obers:2020, von_Blanckenburg.Schwartz:2025}, \cite[Chapter~5]{Schwartz:2026}
\begin{equation}
  M_\mu \xrightarrow{u \to u - v} M_\mu + v_\mu + \tau_\mu \tfrac{v^2}{2}\;.
\end{equation}
Then, any connection constructed from the fields
\begin{equation*}
  \bar h_{\mu\nu} := \ub{}_{\mu\nu} - 2 M_{(\mu} \tau_{\nu)} \; , \\
  \bar u^\mu := u^\mu + M_\alpha h^{\mu\alpha} \; , \\
  \bar \ator_\mu := \Lie{\bar u} \tau_\mu \; , \\
    \bar\Theta_{\mu\nu} := \tfrac{1}{2}\Lie{\bar u} \bar h_{\mu\nu} \; , \\
    \bar\theta := h^{\alpha\beta} \bar\Theta_{\alpha\beta} \; ,
\end{equation*}
which are boost-invariant fields, will be boost-invariant.  In this
sense, the connection is constructed from $(\tau_\mu, h^{\mu\nu}, M_\mu, u^\mu,
\ukappa{}_{\mu\nu})$ and linear in the first derivatives of $(\tau_\mu, h^{\mu\nu},
M_\mu)$.

As an example of such a connection, refs.~\cite{Jensen:2018, Hartong.Kiritsis.Obers:2015, Bekaert.Morand:2016, Festuccia.EtAl:2016, Hansen.Hartong.Obers:2020}
consider the following torsionful compatible connection (note that the
convention used for unit timelike vector fields in refs.~\cite{Jensen:2018, Hartong.Kiritsis.Obers:2015, Festuccia.EtAl:2016, Hansen.Hartong.Obers:2020}
is sign-opposite to ours)
\begin{align}
  \label{eq:con_Mass}
  \bar\Gamma^\alpha_{\mu\nu}
    &:= h^{\alpha\sigma}\left( \partial_{(\mu} \bar h_{\nu)\sigma} - \tfrac{1}{2} \partial_\sigma \bar h_{\mu\nu} \right) + \bar u^\alpha \partial_{\mu} \tau_{\nu} \\
    &\:= \ucheckGamma{\bar u}^\alpha_{\mu\nu}
      + 2 \tau_{(\mu} \ukappa{\bar u}_{\nu)\sigma} h^{\alpha\sigma} + \bar u^\alpha \tau_{[\mu} \ator_{\nu]} \; ,
\end{align}
where
$\ukappa{\bar u}_{\mu\nu}
    = \tau_{[\mu} \ub{\bar u}_{\nu]\alpha}
        \left(\ator^\alpha + h^{\alpha\beta}\partial_\beta\right)
        \left(-2 M_\sigma \bar u^\sigma + M_\sigma M_\gamma h^{\sigma\gamma}\right)
$.

To get the most general connection in this approach, it is sufficient
to replace, in the connection~\eqref{eq:con_Gal_ator}, the fields $\ub{}_{\mu\nu}$, $u^\mu$ and
$\uator{}_\mu$ by, respectively, the boost-invariant fields $\ub{\bar
  u}_{\mu\nu} = \bar h_{\mu\nu} + \tau_\mu \tau_\nu \left(-2 M_\alpha \bar u^\alpha + M^2\right)$,
$\bar u^\mu$ and $\bar \ator_\mu$, and choosing the Coriolis field of
$\bar u^\mu$ to be boost-invariant.\footnote{One could also replace $\ub{}_{\mu\nu}$ by $\bar h_{\mu\nu}$.} (For simplicity, we do not
consider terms involving $\bar\Theta_{\mu\nu}$ or $\bar\theta$.) This leads to
\begin{align}
  \label{eq:con_Gal_ator_Mass}
  \Gamma^\alpha_{\mu\nu}
  &= \ucheckGamma{\bar u}^\alpha_{\mu\nu} + 2\tau_{(\mu} \ukappa{\bar u}_{\nu)}{}^\alpha
    \nonumber\\
  &\quad+ c_1 \ub{\bar u}^\alpha{}_{(\mu} \bar \ator_{\nu)}
    + c_2 \ub{\bar u}_{\mu\nu} \bar \ator^\alpha
    + c_3 \, \bar u^\alpha \tau_{(\mu} \bar \ator_{\nu)}
    + c_4 \ub{\bar u}^\alpha{}_{[\mu} \bar \ator_{\nu]}
    + c_5 \, \bar u^\alpha \tau_{[\mu} \bar \ator_{\nu]} \; .
\end{align}
This connection is directly boost-invariant, and therefore the
parameters $c_i$ are totally free.

The clear disadvantage of this approach is that it requires the
introduction of a new field, i.e., namely the mass gauge field $M_\mu$.
A direct consequence of the presence of this field is that a preferred
timelike vector field is defined, corresponding to $\bar u^{\mu}$.
Therefore this approach is equivalent to considering the connection~\eqref{eq:con_Gal_ator}
and saying that the vector field $u^\mu$ is a fundamental building block
of the theory, and that no boost invariance with respect to that
vector field should be required.  In other words, this approach
follows a stronger version of the approach of \Cref{app:restricted_boosts}, as the Galilean
boosts are restricted to be the trivial boost.

\begin{remark}
  We could construct a connection of the form \eqref{eq:unique_A_con} with a gauge field
  by choosing $\ukappa{\bar u}_{\mu\nu} = 0$, $c_2 = -\tfrac{1}{2}$ and
  $c_1 = c_3 = c_4 = c_5 = 0$ in the connection~\eqref{eq:con_Gal_ator_Mass}.  In this case,
  we have $\ukappa{u}_{\mu\nu} = -\partial_{[\mu} M_{\nu]}$ in the connection~\eqref{eq:unique_A_con}.
\end{remark}

\begin{remark}
  Even if the Coriolis field is fully constructed from the gauge
  field, as for the connection~\eqref{eq:con_Mass}, and therefore the connection is
  constructed only from $(\tau_\mu, h^{\mu\nu}, u^\mu, M_\mu)$, there is always a
  preferred observer that can be constructed from $M_\mu$.  On the
  contrary, if one considers only $(\tau_\mu, h^{\mu\nu}, u^\mu,
  \ukappa{}_{\mu\nu})$, as we do in the present paper, no preferred
  observer can be constructed from $\ukappa{}_{\mu\nu}$ in general.  This
  is why $(\tau_\mu, h^{\mu\nu}, u^\mu, \ukappa{}_{\mu\nu})$ carries less structure
  than $(\tau_\mu, h^{\mu\nu}, u^\mu, M_\mu)$.  The use of the Coriolis field is
  further supported by the classification of affine connections with
  respect to Galilean structures \cite{Schwartz:2025}, which implies that the
  Coriolis field is one of the `building blocks' of an affine
  connection with respect to a Galilean structure.
\end{remark}

\subsubsection{Choice of connection with a mass gauge field}

While the gauge field has the disadvantage of introducing a preferred
timelike vector field, the freedom in the connection~\eqref{eq:con_Gal_ator_Mass} offers several
interesting possibilities for choosing the non-metricities and
torsion, which is not possible with the unique connection~\eqref{eq:unique_A_con}.  In
particular, there are two natural choices of connection that can be
considered:
\begin{enumerate}
\item \emph{A torsionful compatible connection}: this implies $c_2 = - c_1$, $c_3
  = 0$, $c_4 = c_1$, $c_5 = 1$, where $c_1$ is free.  In this case the
  connection is not unique.
\item \emph{A torsion-free conformally compatible connection}: this implies
  $c_1 = 2$, $c_2 = -1$, $c_3 = 1$, $c_4 = 0$, $c_5 = 0$, and the
  connection is uniquely given by
  \begin{subequations}
  \begin{equation}
    \Gamma^\alpha_{\mu\nu}
    = \ucheckGamma{\bar u}^\alpha_{\mu\nu} + 2\tau_{(\mu} \ukappa{\bar u}_{\nu)}{}^\alpha
      + 2 \ub{\bar u}^\alpha{}_{(\mu} \bar\ator_{\nu)}
      - \ub{\bar u}_{\mu\nu} \bar\ator^\alpha
      + \bar u^\alpha \tau_{(\mu} \bar\ator_{\nu)}
  \end{equation}
  or, equivalently, by
  \begin{equation}
    Q_{\mu\nu} = -\bar\ator_\mu \tau_\nu \; , \quad
    Q_\alpha{}^{\mu\nu} = 2\bar\ator_\alpha h^{\mu\nu} \; , \quad
    T^\alpha{}_{\mu\nu} = 0 \; .
  \end{equation}
  \end{subequations}
  Here the conformal 1-form $n_\mu$ in~\eqref{eq:conf_connection} corresponds to
  $-\bar\ator_\mu$.
\end{enumerate}

\subsection{With a preferred time function}

As discussed in \Cref{sec:Galilean_structures} and shown in ref.~\cite{Vigneron.Barzegar.Read:2025}, if the Galilean structure
results from the Galilean limit of a globally hyperbolic Lorentzian
structure, then there exists globally a scalar field $\Psi$ and a time
function $t$ such that $\tau = \e^{-\Psi} \dd t$, which the Frobenius
condition $\tau \wedge \dd \tau = 0$ alone only implies locally.  However, given
$\tau_\mu$, the field $\Psi$ and the time $t$ are not uniquely defined as can
be seen with the transformation $t \rightarrow g(t)$ leading to $\tau = \e^{-\Psi +
  \ln \dot g} \dd t$ and therefore $\Psi \rightarrow \Psi - \ln \dot g$.  This is in
contrast with the classical situation where $\tau_\mu$ is exact, and
therefore a preferred, or \emph{absolute}, time function $\bar t$ is defined
such that $\tau = \dd \bar t$.

In the approach of this section to construct a boost-invariant
connection, we reintroduce a preferred time function $\bar t$, and
related \emph{time dilation field} $\bar\Psi$, such that
\begin{align}
  \tau = \e^{-\bar\Psi} \dd \bar t \; .
\end{align}
Note that, even with the introduction of this `preferred' time
function, the physical time is still measured by $\tau_\mu$, and therefore
the spatial dependence of $\bar\Psi$ induces physical time dilation.  We
define the field $X_\mu := \partial_\mu \bar\Psi$ , with $\ator^\mu = h^{\mu\alpha} X_\alpha$.

\begin{remark}
  While this approach introduces a 1-form field $X_\mu$ as an extra
  structure, similarly to the gauge field approach with $M_\mu$, these
  two 1-forms differ in two ways: (i)~only the temporal part of $X_\mu$
  is independent of $\ator^\mu$ while $M_\mu$ is fully independent of
  $\ator^\mu$, and (ii)~$X_\mu$ is boost-invariant, while $M_\mu$ transforms
  as $M_\mu \xrightarrow{u \to u - v} M_\mu + v_\mu + \tau_\mu \tfrac{v^2}{2}$.
\end{remark}

Therefore, compared to the approach of \Cref{sec:unique_A_con} and the proof of \Cref{app:proof_unique_A_con} that
lead to the unique connection~\eqref{eq:unique_A_con}, on top of $\ator^\mu$, we have the
additional field $u^\mu X_\mu$ from which we can construct a connection.
Then, a connection constructed from $(\bar\Psi, \bar t, h^{\mu\nu}, u^\mu,
\ukappa{}_{\mu\nu})$ and linear in the first derivatives of $(\bar\Psi,
h^{\mu\nu})$ has the general form\footnote{We omit terms involving
  $\Lie{u}\ub{}_{\mu\nu}$ as they must vanish for a boost-invariant
  connection.  Indeed, the presence of $\bar\Psi$ does not change the
  derivation made in \Cref{app:proof_unique_A_con} for these terms, while it does for the terms
  involving $\ator^\mu$.}
\begin{align}
  \label{eq:con_Gal_ator_Psi}
  \Gamma^\alpha_{\mu\nu}
  &= \ucheckGamma{}^\alpha_{\mu\nu} + 2\tau_{(\mu} \ukappa{}_{\nu)}{}^\alpha \nonumber\\
  &\quad+ c_1 \ub{}^\alpha{}_{(\mu} \ator_{\nu)}
    + c_2 \ub{}_{\mu\nu} \ator^\alpha
    + c_3 \, u^\alpha \tau_{(\mu} \ator_{\nu)}
    + c_4 \ub{}^\alpha{}_{[\mu} \ator_{\nu]}
    + c_5 \, u^\alpha \tau_{[\mu} \ator_{\nu]} \nonumber\\
  &\quad+ \left[c_6 \ub{}^\alpha{}_{(\mu} \tau_{\nu)}
    + c_7 \, u^\alpha \ub{}_{\mu\nu}
    + c_8 \, u^\alpha \tau_\mu \tau_\nu
    + c_9 \ub{}^\alpha{}_{[\mu} \tau_{\nu]}\right] \, u^\sigma \partial_\sigma \bar\Psi \; ,
\end{align}
where the $c_i$ are constants.  The non-metricities and torsion of
this connection are
\begin{subequations}
\begin{align}
  Q_{\mu\nu}
  &= -(1 - c_5) \uator{}_{[\mu} \tau_{\nu]}
    - c_3 \uator{}_{(\mu} \tau_{\nu)}
    - (c_7 \ub{}_{\mu\nu} + c_8 \tau_\mu \tau_\nu) \, u^\sigma \partial_\sigma \bar\Psi \; ,
  \\
  Q_\alpha{}^{\mu\nu}
  &= (c_1 - c_4) \uator{}_\alpha \, h^{\mu\nu}
    + (-1 + c_3 + c_5) \, \tau_\alpha u^{(\mu} \ator^{\nu)}
    + (c_1 + 2 c_2 + c_4) \ub{}^{(\mu}{}_\alpha \ator^{\nu)} \nonumber\\
  &\quad + \left[(c_6 - c_9) \, \tau_\alpha h^{\mu\nu}
      + 2 c_7 \ub{}^{(\mu}{}_\alpha \, u^{\nu)} \right] u^\sigma \partial_\sigma \bar\Psi \; , \\
  T^\alpha{}_{\mu\nu}
  &= 2 c_4 \ub{}^\alpha{}_{[\mu} \uator{}_{\nu]}
    + 2 c_5 \, u^\alpha \tau_{[\mu} \uator{}_{\nu]}
    + 2 c_9 \, \delta^\alpha_{[\mu} \tau^{\vphantom{\alpha}}_{\nu]} \, u^\sigma \partial_\sigma \bar \Psi \; .
\end{align}
\end{subequations}

With a calculation similar to that performed in \Cref{app:proof_unique_A_con}, we can show that
the connection~\eqref{eq:con_Gal_ator_Psi} is boost-invariant if and only if
\begin{equation}
  c_3 = 1 + c_1 + 2 c_2 \; , \quad
  c_5 = c_4 \; , \quad
  c_6 = c_1 \; , \quad
  c_7 = 0 \; , \quad
  c_8 = 1 + c_1 + 2 c_2 \; , \quad
  c_9 = c_4 \; ,
\end{equation}
leading to
\begin{align}
  \Gamma^\alpha_{\mu\nu}
  &= \ucheckGamma{}^\alpha_{\mu\nu} + 2\tau_{(\mu} \ukappa{}_{\nu)}{}^\alpha \nonumber \\
  &\quad+ c_1 \ub{}^\alpha{}_{(\mu} X_{\nu)}
    + c_2 \ub{}_{\mu\nu} X^\alpha
    + (1 + c_1 + 2 c_2) \, u^\alpha \tau_{(\mu} X_{\nu)}
    + c_4 \, \delta^\alpha_{[\mu} X^{\vphantom{\alpha}}_{\nu]} \; ,
\end{align}
and
\begin{subequations}
\begin{align}
  Q_{\mu\nu}
  &= (1 - c_4) \, \tau_{[\mu} X_{\nu]}
    - (1 + c_1 + 2 c_2 ) X_{(\mu} \tau_{\nu)} \; , \\
  Q_\alpha{}^{\mu\nu}
  &= (c_1 - c_4) \, X_\alpha h^{\mu\nu}
    + (c_1 + 2 c_2 + c_4) \delta^{(\mu}_\alpha X^{\nu)}_{\vphantom{\alpha}} \; , \\
  T^\alpha{}_{\mu\nu}
  &= 2 c_4 \, \delta^\alpha_{[\mu} X^{\vphantom{\alpha}}_{\nu]} \; .
\end{align}
\end{subequations}
The non-metricities and torsion are boost-invariant since $X_\mu$ is
boost-invariant.

While we introduce a preferred time function from which the clock form
is defined, the boost-invariant connection we construct in the present
approach does not introduce a preferred timelike vector field, unlike
the restricted boost approach from \Cref{app:restricted_boosts} or the mass gauge field
approach from \Cref{app:con_Mass}.  Additionally, since a preferred time is also
defined in the classical case, for which $\tau = \dd t$, but still no
preferred observer, this makes the `preferred-time-function approach'
the most natural approach to construct a boost-invariant connection
after the one considered in \Cref{sec:unique_A_con}.

\begin{remark}
  To get a connection that is independent of the time $\bar t$, i.e.,
  independent of the transformation $\bar\Psi \rightarrow \bar\Psi + f(t)$, the only
  possibility is $c_6 = c_7 = c_8 = c_9$, leading again to the unique
  boost-invariant connection~\eqref{eq:unique_A_con}.
\end{remark}

\subsubsection{Choice of connection with a preferred time function}

As for the mass gauge field approach, the freedom in the connection
offers several possibilities for choosing the non-metricities and
torsion:
\begin{enumerate}
\item \emph{A torsionful compatible connection}: this implies $c_1 = 1$, $c_2 =
  -1$ and $c_4 = 1$, and the connection is uniquely given by
  \begin{subequations}
  \begin{equation}
    \Gamma^\alpha_{\mu\nu}
    = \ucheckGamma{}^\alpha_{\mu\nu} + 2\tau_{(\mu} \ukappa{}_{\nu)}{}^\alpha
      + \ub{}^\alpha{}_{(\mu} X_{\nu)}
      - \ub{}_{\mu\nu} X^\alpha
      + \ub{}^\alpha{}_{[\mu} X_{\nu]}
      + \, u^\alpha \tau_{[\mu} X_{\nu]}
  \end{equation}
  or, equivalently, by
  \begin{equation}
    Q_{\mu\nu} = 0 \; , \quad
    Q_\alpha{}^{\mu\nu} = 0 \; , \quad
    T^\alpha{}_{\mu\nu} = 2 \delta^\alpha_{[\mu} X^{\vphantom{\alpha}}_{\nu]} \; .
  \end{equation}
  \end{subequations}
\item \emph{A torsion-free conformally compatible connection}: this implies $c_1
  = 2$, $c_2 = -1$ and $c_4 = 0$, and the connection is uniquely
  given by
  \begin{subequations}
  \begin{equation}
    \Gamma^\alpha_{\mu\nu}
    = \ucheckGamma{u}^\alpha_{\mu\nu} + 2\tau_{(\mu} \ukappa{u}_{\nu)}{}^\alpha
      + 2 \ub{u}^\alpha{}_{(\mu} X_{\nu)}
      - \ub{u}_{\mu\nu} X^\alpha
      + u^\alpha \tau_{(\mu} X_{\nu)}
  \end{equation}
  or, equivalently, by
  \begin{equation}
    Q_{\mu\nu} = - X_\mu \tau_\nu \; , \quad
    Q_\alpha{}^{\mu\nu} = 2 X_\alpha h^{\mu\nu} \;, \quad
    T^\alpha{}_{\mu\nu} = 0 \; .
  \end{equation}
  \end{subequations}
  As for the mass gauge field approach, the conformal 1-form $n_\mu$
  in~\eqref{eq:conf_connection} corresponds to $-\bar\ator_\mu$.
\end{enumerate}

\printbibliography[heading=bibintoc]

\end{document}